%% file: main.tex
\documentclass[11pt,letterpaper]{article}

\usepackage[utf8]{inputenc}
\usepackage[T1]{fontenc}
\usepackage{amsmath,amssymb,amsthm}
\usepackage{mathtools}
\usepackage{braket}
\usepackage[margin=1in]{geometry}
\usepackage{hyperref}
\usepackage{aliascnt}
\usepackage{cleveref}
\usepackage{enumitem}
\usepackage{booktabs}
\usepackage{algorithm}
\usepackage{algpseudocode}

\newtheorem{theorem}{Theorem}[section]

\newaliascnt{lemma}{theorem}
\newtheorem{lemma}[lemma]{Lemma}
\aliascntresetthe{lemma}
\crefname{lemma}{Lemma}{Lemmas}
\Crefname{lemma}{Lemma}{Lemmas}

\newaliascnt{corollary}{theorem}
\newtheorem{corollary}[corollary]{Corollary}
\aliascntresetthe{corollary}
\crefname{corollary}{Corollary}{Corollaries}
\Crefname{corollary}{Corollary}{Corollaries}

\theoremstyle{definition}
\newaliascnt{definition}{theorem}
\newtheorem{definition}[definition]{Definition}
\aliascntresetthe{definition}
\crefname{definition}{Definition}{Definitions}
\Crefname{definition}{Definition}{Definitions}

\theoremstyle{plain}
\newaliascnt{proposition}{theorem}
\newtheorem{proposition}[proposition]{Proposition}
\aliascntresetthe{proposition}
\crefname{proposition}{Proposition}{Propositions}
\Crefname{proposition}{Proposition}{Propositions}

\newaliascnt{conjecture}{theorem}

\aliascntresetthe{conjecture}
\crefname{conjecture}{Conjecture}{Conjectures}
\Crefname{conjecture}{Conjecture}{Conjectures}

\theoremstyle{definition}
\newaliascnt{assumption}{theorem}
\newtheorem{assumption}[assumption]{Assumption}
\aliascntresetthe{assumption}
\crefname{assumption}{Assumption}{Assumptions}
\Crefname{assumption}{Assumption}{Assumptions}

\theoremstyle{remark}
\newaliascnt{remark}{theorem}

\aliascntresetthe{remark}
\crefname{remark}{Remark}{Remarks}
\Crefname{remark}{Remark}{Remarks}

\newtheoremstyle{informal}{\topsep}{\topsep}{\itshape}{}{\scshape}{.}{.5em}{\thmname{#1}\thmnumber{ #2}\thmnote{ (#3)}}
\theoremstyle{informal}
\newtheorem{informaltheoremhead}{Informal Theorem}
\theoremstyle{plain}
\newenvironment{informaltheorem}[1][]{\begin{informaltheoremhead}[#1]\leavevmode}{\end{informaltheoremhead}}

\newcommand{\condH}{\textup{(H)}}
\newcommand{\poly}{\mathrm{poly}}
\newcommand{\negl}{\mathrm{negl}}
\newcommand{\BQP}{\mathsf{BQP}}
\newcommand{\NP}{\mathsf{NP}}

\newcommand{\Tr}{\mathrm{Tr}}
\newcommand{\calA}{\mathcal{A}}
\newcommand{\calC}{\mathcal{C}}
\newcommand{\calF}{\mathcal{F}}

\newcommand{\calR}{\mathcal{R}}
\newcommand{\calB}{\mathcal{B}}
\newcommand{\Meas}{\mathsf{M}}
\newcommand{\Mint}{\mathsf{Mint}}
\newcommand{\Ver}{\mathsf{Ver}}
\newcommand{\accept}{\mathsf{acc}}
\newcommand{\reject}{\mathsf{rej}}
\newcommand{\openone}{\mathbb{I}}
\newcommand{\Ima}{\mathrm{Im}}
\newcommand{\openidws}{\mathbb{I}_{\mathrm{anc}}}

\title{Anchoring the Security of Invariant Quantum Money}
\author{
  \begin{tabular}{cc}
    Hans Schmiedel & Jiangshan Yu \\
    \small\textit{The University of Sydney} & \small\textit{The University of Sydney} \\
    \small\textit{hans.mailbox@tutamail.com} & \small\textit{jiangshan.yu@sydney.edu.au}
  \end{tabular} 
  }
\date{November 2, 2026}

\providecommand{\Sym}{\mathrm{Sym}}
\providecommand{\Prep}{\mathsf{Prep}}
\providecommand{\aw}{\mathrm{aw}}
\providecommand{\win}{\mathrm{win}}
\providecommand{\adv}{\mathrm{adv}}
\providecommand{\cnt}{\mathrm{cnt}}
\providecommand{\Syn}{\mathsf{Syn}}
\providecommand{\Sec}{\mathsf{Sec}}
\providecommand{\plusS}{\ket{+_{2|S|}}}
\begin{document}
\maketitle
\input{abstract}
\clearpage

\input{introduction}

\input{body}

\input{conclusion}
\input{acknowledgements}

\input{appendix}

\end{document}

%% file: abstract.tex
\begin{abstract}
Prior studies on quantum money from knots, and the invariant money framework that generalizes it, anchor their security on the hardness of \emph{path-finding}. However, we show that the hardness of path-finding is proven to be necessary for security under certain conditions. In other conditions, this assumption is not a direct anchoring hardness problem such that breaking it allows an adversary to forge a banknote directly. In particular, for quantum money from knots, even when the path-finding assumption does not hold, an adversary still cannot directly forge a banknote. The adversary needs to additionally prepare the uniform superposition over the paths it finds, which we show is as hard as the original problem. 

Existing security proofs supplement path-finding hardness with knowledge-of-path assumptions. While this combination is sufficient, it may not be necessary. We identify an assumption pair, \emph{preparation assumption} and \emph{transfer assumption}, that is both necessary and sufficient for the security of every invariant money scheme. Breaking either assumption directly enables a forgery. The \emph{preparation assumption} says that no efficient algorithm, given the serial number of a banknote and one object measured from it, prepares a state that verification accepts. While quantum money security gives a forger one quantum state and asks for two that pass, our assumption concerns a problem that takes only classical inputs and asks for one quantum state that passes. This assumption is equivalent to security against forgers that measure their banknote first. Given that its inputs are classical, it is falsifiable and can be attacked and bounded like other computational problems. The \emph{transfer assumption} covers forgers that keep the banknote coherent, stating that every such forger only has an advantage over the measure-first adversary by at most a polynomial factor. Together, these assumptions are equivalent to security.
\end{abstract}

%% file: introduction.tex
\section{Introduction}\label{sec:intro}

Quantum money uses quantum states as banknotes that cannot be copied:
classical data can always be duplicated, but an unknown quantum state
cannot \cite{WZ82,Dieks82}. A scheme has a \emph{mint}, which produces a
banknote, a quantum state labelled by a classical serial number, and a
\emph{verifier}, which accepts or rejects a state at a given serial
number. The scheme is public-key if anyone can run the verifier
\cite{Aaronson09}, and it is secure if no efficient forger can turn one
banknote into two states that both pass.
 
The schemes we study follow the blueprint of quantum money from knots
\cite{FGHLS10}, which Liu, Montgomery and Zhandry abstracted into the
framework of \emph{invariant money} \cite{LMZ22}. Fix a set $X$ of
combinatorial objects, a set of \emph{moves}, each an efficiently
computable permutation of $X$, and an efficiently computable
\emph{invariant} $I$ that every move preserves; for knots, these are
grid diagrams, grid moves, and the Alexander polynomial respectively. The mint
prepares the uniform superposition over $X$, computes the invariant
into a second register, and measures it:
\begin{equation}\label{eq:blueprint}
  \frac{1}{\sqrt{|X|}}\sum_{x\in X}\ket{x}\ket{0}
  \ \xrightarrow{\ \text{compute }I\ }\
  \frac{1}{\sqrt{|X|}}\sum_{x\in X}\ket{x}\ket{I(x)}
  \ \xrightarrow{\ \text{measure}\ }\
  \ket{\$_p}\ket{p},
\end{equation}
where $\ket{\$_p}$ is the uniform superposition over the objects with
invariant $p$. The outcome $p$ is the serial number and $\ket{\$_p}$
the banknote. The verifier checks that the invariant is $p$, then
applies a random move in superposition and tests whether the state is
left unchanged; a banknote always is, because the moves only permute
the objects with invariant $p$. The objects reachable from one
another by moves form an \emph{orbit}, and the uniform superposition
over an orbit, its \emph{orbit state}, is what passes this test with
certainty; a banknote is a superposition of orbit states.  Verification repeats this test, so its random moves perform a random walk on each orbit, the \emph{verification walk}.  The test is meaningful when it is repeated long enough for this walk to spread over the whole orbit.
 
What stops a forger? \cite{FGHLS10} base the security of the knot
scheme on the hardness of \emph{path-finding}: given two objects of
one orbit, exhibit a sequence of moves from one to the other. The
guess is natural, since the moves are what define an orbit, and they
support it with a scheme on graphs, where the moves between two
objects are easy to find in practice, and which they break. An algorithm that efficiently solves path-finding, a \emph{path-finder}, answers a classical question
about two given objects, while a forger must produce a quantum state
spread over a whole orbit. On graphs the first gives the second. For
knots it is not clear what a path-finder gives: an efficient path-finder for grid diagrams offers no evident way to produce a second banknote. So hardness of path-finding may not be necessary for the security of the knot scheme, and security has not been proved from it alone. Existing works~\cite{LMZ22,MS24} add knowledge-of-path assumptions: any efficient algorithm that produces two objects with the same invariant is assumed to know, implicitly, a path between them. Such an assumption may be stronger than security itself, so a proof from it may assume more than it proves.  It may even be false for a secure scheme: Zhandry shows that its direct analogue for group actions is false~\cite[\S1.3]{Zhandry24}.

We say that an assumption \emph{anchors} the security of a quantum money scheme if it concerns a problem with classical inputs and breaking it directly allows an efficient adversary to forge a banknote.

We ask the following:
\begin{enumerate}
\item[\textbf{Q1.}] Under which conditions is hardness of path-finding necessary for security, and can security be proved from it?
\item[\textbf{Q2.}] What assumption anchors the security of every invariant money scheme?
\end{enumerate}

\subsection{Results}

We answer both questions with the following results.

\paragraph{When path-finding anchors security.} (\Cref{sec:where-not})
What a path-finder gives depends on structure that graphs have and that has not been found for knots.  The attack of \cite[\S3.4]{FGHLS10} on graphs uses an algorithm that finds graph isomorphisms, which for graphs is a path-finder (\Cref{sec:coset}), together with the relabellings of the vertices.  These form a group that is easy to superpose over, acts transitively on each orbit, and contains every move, with its operations computable in polynomial time.  We say that a scheme satisfies \emph{condition~\condH} when its orbits come with such a group (\Cref{def:condH}), and we show that the attack works for every scheme that satisfies it.  Informal Theorem~\ref{inf:pf} answers the first half of Q1 with condition~\condH, and shows what changes without it.
\begin{informaltheorem}[Necessity of path-finding hardness under condition~\condH]\label{inf:pf}
\begin{enumerate}[label=\textup{(\roman*)},leftmargin=*,nosep]
\item If condition~\condH{} holds, an adversary with a path-finder can forge a banknote.
\item If condition~\condH{} does not hold, an adversary with a path-finder needs to additionally prepare the uniform superposition over the paths it finds to forge a banknote, which is as hard as preparing the orbit state.

\end{enumerate}
\end{informaltheorem}
Part~(i) makes the intuition behind the knot scheme precise wherever condition~\condH{} holds: there, hardness of path-finding anchors security, although security has not been proved from it.  The construction behind it also gives concrete attacks.  Montgomery and Sharif note that solving the group-action discrete logarithm problem would let an adversary forge their class-group scheme~\cite[\S2.2]{MS24}.  \Cref{prop:ms-forgery} makes this explicit: the best known algorithms for that problem~\cite{CJS14} give a quantum subexponential forgery, assuming the generalized Riemann hypothesis.  No group satisfying condition~\condH{} has been found for knot orbits, and part~(ii) answers the question raised above: there, an efficient path-finder leaves the attacker a preparation problem as hard as the original one.  So hardness of path-finding may not anchor the security of the knot scheme.

\paragraph{Orbit state preparation anchors security.} (\Cref{sec:osp-problem,sec:reduction})
Hardness of path-finding thus anchors security when condition~\condH{} holds, but it may not anchor the security of invariant money in general. We look instead at what every forger must do, with or without a path-finder: prepare a state that verification accepts.  A forger either measures its banknote, after which it holds only the serial number $p$ and the object $x$ it measured, or keeps the banknote coherent. We give one assumption for each case, and show that security is equivalent to the two together.\footnote{This and the results below assume that the scheme \emph{mixes rapidly}: verification is repeated long enough for its walk to spread over each orbit, as described above (\Cref{def:rapid-mixing}).  Without this, verification can accept a state that is orthogonal to every banknote.}  In the first case, when the invariant takes different values on different orbits, the forger faces \emph{orbit state preparation}: given $x$, prepare the orbit state of the orbit of $x$ (\Cref{def:superposition}).  \cite[\S1.2, \S3.3]{FGHLS10} describe this approach to forgery informally, and Lutomirski studies the problem, as \textsc{Component Superposition}, in a black-box model of the moves~\cite{Lutomirski11}.  For this case we introduce the \emph{preparation assumption} (\Cref{ass:synth}): given the serial number $p$ of a minted banknote and an object $x$ measured from it, no efficient algorithm can output a state that verification accepts at $p$.  For the second case we introduce the \emph{transfer assumption} (\Cref{ass:transfer}): keeping the banknote coherent helps a forger by at most a polynomial factor in its success probability.
\begin{informaltheorem}[Security decomposition of invariant quantum money]\label{inf:decomposition}
\begin{enumerate}[label=\textup{(\alph*)},leftmargin=*,nosep]
\item A forger that first measures its banknote can succeed exactly when the preparation assumption is broken.
\item A scheme can be forged exactly when the preparation assumption or the transfer assumption is broken.
\end{enumerate}
\end{informaltheorem}
The preparation assumption anchors security: breaking it breaks the scheme, by part~(a), and its instances $(p,x)$ are classical, so every proof of security must establish it.  It is moreover falsifiable in the sense of Naor (\Cref{prop:falsifiable}), and it can be attacked and bounded like other computational problems: an $\NP$ oracle solves orbit state preparation (\Cref{prop:np-forges}), so the assumption implies $\NP\not\subseteq\BQP$.  Part~(b) says that the two assumptions together capture exactly what security says, whereas a proof from knowledge of path may assume more.  The transfer assumption is the residual: breaking it also breaks the scheme, but unlike the preparation assumption it comes with no efficient test that could refute it, and \Cref{thm:transfer-barrier} shows that a natural class of reductions cannot derive it from the preparation assumption.

\paragraph{A barrier to deriving the transfer assumption.} (\Cref{sec:transfer-necessity})
The best outcome would be to derive the transfer assumption from the preparation assumption: security would then follow from a single falsifiable assumption with classical inputs.  We consider reductions, which use a forger as a subroutine to break the underlying assumption.  We show that, if the preparation assumption holds, the natural kind of reduction cannot derive the transfer assumption (\Cref{thm:transfer-barrier}).  Such a reduction starts from the classical data $(p,x)$ of the preparation assumption, runs the forger only at the serial number $p$, and treats it as a black box, although it may run it in superposition, rewind it, and choose all of its inputs, including its workspace.  This includes any black-box attempt to extract a path from the forger, the style of argument behind knowledge of path.  The obstacle is the coherence of the banknote.  Consider an all-powerful forger that outputs a second banknote whenever its input is close to a genuine one, and otherwise does nothing.  It always wins the security game, yet a reduction that holds only $(p,x)$ cannot use it: to get anything out of it, the reduction would first have to prepare a state that verification accepts, which is exactly what the preparation assumption rules out.  A proof of security from the preparation assumption must therefore leave this class.  Our argument adapts the separation technique of Gentry and Wichs~\cite{GW11}, who showed that no black-box reduction proves succinct non-interactive arguments sound from a falsifiable assumption.

\paragraph{Design conditions.} (\Cref{sec:verification})
Two simple properties are necessary for security, since a forger wins when either fails (\Cref{prop:design-necessary}): the mint must not issue any one serial number too often, and a minted banknote must lie on a large orbit.  A third property, rapid mixing, says that verification repeats its random moves long enough for them to spread over each orbit, so that a state passes verification only if it is close to a superposition of orbit states.  Together, these are properties that a designer can check without reasoning about adversaries.  Even all three together do not imply security: a simple scheme whose orbits are shifted copies of a linear subspace has all three, yet its orbit states are easy to prepare (\Cref{prop:design-insufficient}).  Security needs the hardness of orbit state preparation itself.

\subsection{Related work}

\paragraph{Public-key quantum money.}
Public-key quantum money is known from indistinguishability obfuscation and one-way functions~\cite{Zhandry21}.  Zhandry also introduces \emph{quantum lightning}, a stronger notion in which the adversary receives no banknote and must produce two banknotes for a serial number of its own choice~\cite{Zhandry21}, and lightning follows from the one-shot signatures of Shmueli and Zhandry, built from subexponentially secure obfuscation and learning with errors~\cite{SZ25}.  Candidates that avoid obfuscation have fared worse.  The first scheme, due to Aaronson~\cite{Aaronson09}, was broken by Lutomirski et al.~\cite{Lutomirski10}.  The hidden-subspace scheme of Aaronson and Christiano is secure relative to a classical oracle~\cite{AC12}, but its instantiation by low-degree polynomials was attacked classically~\cite{CPFP15,CPDDF19} and then broken by a polynomial-time quantum algorithm~\cite{BDG23}.  Schemes from modular forms and from quaternion algebras~\cite{Kane18,KSS21} rest on problems specific to each.  Invariant money, which includes the knot scheme and the class-group scheme of Montgomery and Sharif, is one of the few remaining families of candidates outside obfuscation.  Where prior work builds or attacks individual schemes, we give the whole family a single target: the preparation assumption, which every proof of security must establish and every attack can aim at.

\paragraph{Invariant money and knowledge of path.}
Liu, Montgomery and Zhandry introduce invariant money and prove it secure, indeed lightning secure, under a knowledge-of-path assumption~\cite{LMZ22}.  Montgomery and Sharif instantiate the framework with class group actions on elliptic curves and prove security under an isogeny assumption and a knowledge assumption of the same kind~\cite{MS24}.  Our decomposition differs in what it assumes.  The preparation and transfer assumptions together are equivalent to security, so a proof from them assumes nothing that security does not already imply, whereas knowledge of path may be stronger than security.  The preparation assumption is moreover falsifiable, and the barrier explains why a second assumption remains.

\paragraph{Group-action money.}
In group-action money, the banknote for a serial number $h$ is a superposition over an orbit of a cryptographic group action with phases that depend on $h$, and verification recovers $h$ by phase kickback~\cite{Zhandry24}.  Zhandry proves lightning security for abelian group actions in the generic group-action model, under a strengthening of the discrete logarithm assumption~\cite{Zhandry24}.  For groups of smooth order, Doliskani shows that in the same model breaking the scheme is equivalent to computing discrete logarithms~\cite{Doliskani25}.  Bostanci, Nehoran and Zhandry extend the construction to non-abelian group actions and prove lightning from the pre-action security of the group action, without an idealized model~\cite{BNZ25}.  They also build quantum money from one-way homomorphisms, for which cloning a banknote is as hard as preparing the uniform superposition over the image of the homomorphism~\cite{BNZ25}.  Because their banknotes carry phases, these schemes are not invariant money, and condition~\condH{} does not apply to them.  The discrete logarithm plays the role there that path-finding plays here.  For Zhandry's scheme its hardness is equivalent to security, in the generic model and for groups of smooth order, while for invariant money under condition~\condH{} we prove the analogue of one direction: an efficient path-finder breaks the scheme.  No group satisfying condition~\condH{} has been found for knot orbits.

\paragraph{Cloning and telegraphing.}
A state is \emph{telegraphable} if it can be reconstructed from classical information obtained by measuring it.  Quantum fire is a samplable family of states that can be cloned efficiently but not telegraphed efficiently.  Nehoran and Zhandry construct it relative to a unitary oracle~\cite{NZ24}, and later work gives candidates relative to classical oracles~\cite{BNZ25,CGS25}.  A forger that measures its banknote first is a telegrapher of a particular kind: it must rebuild accepted states from the outcome of a measurement in the computational basis.  The preparation assumption says that such telegraphing fails, and the transfer assumption says that, for invariant banknotes, cloning is at most polynomially easier than telegraphing in this way: invariant banknotes do not behave like quantum fire.  Quantum fire shows that this can fail for other families of states, at least relative to oracles, So the transfer assumption cannot hold for generic reasons: a proof of it must use the structure of invariant banknotes.

\paragraph{The complexity of preparing states.}
Orbit state preparation is a state synthesis problem with a classical input, in the sense studied by Aaronson~\cite{Aaronson16}: given an object, prepare a specific quantum state.  Our upper bound, which prepares orbit states with an $\NP$ oracle (\Cref{prop:np-forges}), is an instance of Aaronson's state synthesis algorithm~\cite{Aaronson16,INNRY22}.  Rosenthal and Yuen~\cite{RY22} and Metger and Yuen~\cite{MY23} develop complexity classes for state synthesis, and show that the states preparable in polynomial space are exactly those that a polynomial-time verifier can prepare with the help of an untrusted prover.  Bostanci et al.\ extend this to unitary synthesis and relate the security of falsifiable quantum cryptographic assumptions to the complexity of the Uhlmann transformation problem~\cite{BEMPQY26}.  The preparation assumption is a cryptographic counterpart of these questions: an average-case state synthesis problem with classical input, whose hardness is equivalent to security against forgers that measure their banknote first.

\paragraph{Falsifiability and black-box barriers.}
Naor's notion of a falsifiable assumption~\cite{Naor03} separates assumptions that an efficient attack could refute from those, such as knowledge assumptions, that no attack can.  Our barrier belongs to the line of Gentry and Wichs~\cite{GW11} discussed above: it limits the proofs that a fixed scheme is secure under a fixed falsifiable assumption.  Black-box separations in quantum cryptography have also concerned quantum money, but as constructions of one primitive from another.  Ananth, Hu and Yuen show that collision-resistant hash functions cannot be used as a black box to build public-key quantum money whose verification makes classical queries~\cite{AHY23}.  {\c{C}}akan et al.\ rule out fully black-box constructions of quantum lightning from obfuscation and one-way functions, or from quantum money, relative to a unitary oracle~\cite{CKMSTY25}.  These results say which primitives can yield quantum money and lightning, and ours says what a proof of security for invariant money from the preparation assumption must do.

%% file: body.tex
\section{Preliminaries and security model}\label{sec:prelim}

This section fixes notation, recalls the framework of invariant money of Liu, Montgomery and Zhandry~\cite[\S8]{LMZ22}, and states the two security games.

\paragraph{Notation.}
A function $\mu$ is \emph{negligible}, written $\mu(n) = \negl(n)$, if $\mu(n) < 1/q(n)$ for every polynomial $q$ and all large $n$.  We write $\poly(n)$ for an unspecified polynomial.
For vectors, $\|\cdot\|$ is the Euclidean norm; for operators, $\|\cdot\|$ is the operator norm and $\|\cdot\|_{\mathrm{tr}}$ the trace norm, the \emph{trace distance} between states $\rho$ and $\rho'$ is $\tfrac12\|\rho-\rho'\|_{\mathrm{tr}}$, and $A \succeq B$ means that $A - B$ is positive semidefinite.
For a finite set $F$ whose elements are named by basis states, $\ket{F} := |F|^{-1/2}\sum_{a \in F}\ket{a}$ is the uniform superposition over $F$.
A pure state on two registers is a \emph{product state}, and the registers are \emph{unentangled}, if it equals $\ket{\alpha}\otimes\ket{\beta}$ for some $\ket\alpha,\ket\beta$, equivalently if its reduced state on either register is pure.

\subsection{Invariant quantum money}\label{sec:framework}

We define the schemes, then minting and verification, and then state how two schemes in the literature fit the definition.

\begin{definition}[Invariant money scheme]\label{def:scheme}
An \emph{invariant money scheme} is a family $(X,S,I) = (X_n,S_n,I_n)_{n\geq1}$, together with a polynomial round count $r = r(n) \geq 1$ of verification (\Cref{alg:ver}), with the following three components.
\begin{enumerate}[label=\textup{(\roman*)},leftmargin=*,nosep]
\item The \emph{ground set} $X \subseteq \{0,1\}^n$ consists of descriptions of combinatorial objects, which we identify with the elements of $X$, and the uniform superposition $\ket X$ is prepared by a uniform family of polynomial-size circuits.
\item The \emph{moves} $S$ are permutations of $X$ with $S = S^{-1}$, indexed as $s_1,\dots,s_{|S|}$ so that $s_i(x)$ and $s_i^{-1}(x)$ are computable in polynomial time from $(i,x)$.
\item The \emph{invariant} $I$ is a function on $X$, computable in polynomial time, with $I(s(x)) = I(x)$ for every $s \in S$ and $x \in X$.
\end{enumerate}
\end{definition}

\paragraph{Words, orbits and level sets.}
A \emph{word} is a finite sequence $w = s_\ell\cdots s_1$ of moves, acting on $X$ by $w\cdot x := s_\ell(\cdots s_1(x)\cdots)$.  We write $S^{\leq L}$ for the set of words of length at most $L$.
The \emph{orbit} of $x$ is $\{w\cdot x : w \text{ a word}\}$, and $x \sim y$ means that $y$ lies in the orbit of $x$.
Orbits partition $X$ and are the connected components of the \emph{move graph}, whose vertices are the elements of $X$ and whose edges join $x$ to $s(x)$ for $s \in S$.
Orbits are indexed by $k$ and written $\calC_k$, and $k(x)$ is the index of the orbit containing $x$.
By requirement~(iii), $I$ is constant on each orbit.
The values of $I$ are the \emph{serial numbers}, and
\[
    P_p \;:=\; \{x \in X : I(x) = p\}
\]
is the \emph{level set} of $p$, a union of orbits.  Distinct orbits may share a serial number.
We write $\{\ket x\}_{x\in X}$ for the corresponding orthonormal basis of $\mathbb{C}^X$, and $\mathsf X$ for the \emph{element register}, which holds a state of $\mathbb{C}^X$: the mint outputs it and the verifier tests it.

\paragraph{Minting.}
$\Mint(1^n)$ prepares the uniform superposition over $X$ with the invariant computed in an adjacent register,
\begin{equation}
    \ket{\Phi} \;=\; \frac{1}{\sqrt{|X|}} \sum_{x \in X} \ket{x}\ket{I(x)}\,,
\end{equation}
measures the second register, and outputs the outcome $p$ together with the residual state
\begin{equation}\label{eq:money-generic}
    \ket{\$_p} \;=\; \frac{1}{\sqrt{|P_p|}} \sum_{x \in P_p} \ket{x}\,.
\end{equation}
The serial number $p$ is published, and $\ket{\$_p}$ is the \emph{banknote}.

\paragraph{Verification.}
A banknote is left unchanged when a move drawn at random is applied to it in superposition, because every move permutes $P_p$.  A state concentrated on one element of $P_p$ is moved elsewhere, and the verifier tests for that difference.  It uses a \emph{move register} $\mathsf A$ with $2|S|$ basis states $\ket 1,\dots,\ket{2|S|}$.  The indices $i \leq |S|$ name the moves $s_i$, and the other $|S|$ indices name the identity, so that half the draws do nothing.  Write
\[
    \plusS \;:=\; \frac{1}{\sqrt{2|S|}}\sum_{i=1}^{2|S|}\ket i\,, \qquad
    U \;:=\; \sum_{i \leq |S|} \ket{i}\!\bra{i} \otimes \hat P_{s_i} \;+\; \sum_{i > |S|} \ket{i}\!\bra{i} \otimes \openone
\]
for the uniform superposition over indices and for the unitary that applies the indexed move, where $\hat P_s$ is the permutation matrix of $s$ on $\mathbb C^X$.  \Cref{alg:ver} gives $\Ver(p,\rho)$, which takes a serial number $p$ and a state $\rho$ on $\mathsf X$, returns $\accept$ or $\reject$, and on acceptance also returns $\mathsf X$.  It is the verification procedure of Farhi et al.~\cite[Alg.~2]{FGHLS10}, with the identity indices that Liu, Montgomery and Zhandry adjoin in their analysis of it~\cite[\S8.2]{LMZ22}.

\begin{algorithm}[!ht]
\caption{$\Ver(p,\rho)$, for a scheme with round count $r$}\label{alg:ver}
\begin{algorithmic}[1]
\State \emph{Serial check.}  Compute $I$ coherently from $\mathsf X$ into a fresh register and measure it.  If the outcome is not $p$, \textbf{reject}.
\State Prepare the move register $\mathsf A$ in $\plusS$.
\For{$t = 1,\dots,r$} \Comment{\emph{mixing test}}
    \State Apply $U$ to $\mathsf A\mathsf X$.
    \State Measure $\mathsf A$ with $\{\ket{+_{2|S|}}\!\bra{+_{2|S|}},\ \openone - \ket{+_{2|S|}}\!\bra{+_{2|S|}}\}$.  On the second outcome, \textbf{reject}.
\EndFor
\State \textbf{Accept}, and return $\mathsf X$.
\end{algorithmic}
\end{algorithm}

On acceptance, the serial check applies to $\mathsf X$ the projector $\Pi_{P_p}$ onto $\mathrm{span}\{\ket{x} : x \in P_p\}$.  One round of the mixing test acts on a pure state $\ket\phi$ of $\mathsf X$ as follows.  Applying $U$, and then projecting $\mathsf A$ onto $\plusS$, gives
\begin{align*}
    U\bigl(\plusS\ket\phi\bigr) &\;=\; \frac{1}{\sqrt{2|S|}}\Bigl(\sum_{i\leq|S|}\ket i\,\hat P_{s_i}\ket\phi \;+\; \sum_{i>|S|}\ket i\,\ket\phi\Bigr)\,,\\
    \bigl(\ket{+_{2|S|}}\!\bra{+_{2|S|}}\otimes\openone\bigr)\,U\bigl(\plusS\ket\phi\bigr) &\;=\; \plusS\otimes\hat B\ket\phi\,,
\end{align*}
since $\braket{+_{2|S|}}{i} = (2|S|)^{-1/2}$ for every index, where
\begin{equation}\label{eq:markov-generic}
    \hat{B} \;:=\; \frac{1}{2}\Bigl(\openone \;+\; \frac{1}{|S|} \sum_{s \in S} \hat{P}_s\Bigr)
\end{equation}
is the \emph{lazy} walk on the move graph: with probability $\tfrac12$ it applies a move drawn uniformly from $S$, and otherwise it leaves the element unchanged.  A round therefore accepts with probability $\|\hat B\ket\phi\|^2$, and on acceptance it returns $\mathsf A$ to $\plusS$, unentangled from $\mathsf X$, so one move register serves all $r$ rounds.  Since $\hat P_s^{T} = \hat P_{s^{-1}}$ and $s\mapsto s^{-1}$ permutes $S = S^{-1}$, $\hat B^{T} = \tfrac12\bigl(\openone + |S|^{-1}\sum_{s\in S}\hat P_{s^{-1}}\bigr) = \hat B$, so $\hat B$ is real symmetric.  The accepting branch of $\Ver(p,\cdot)$ has Kraus operator $\hat B^r\Pi_{P_p}$, and the acceptance probability is
\begin{equation}\label{eq:acc-def}
    \mathrm{Acc}_p(\rho) \;:=\; \Pr[\Ver(p,\rho) = \accept] \;=\; \Tr\bigl(\Pi_{P_p}\hat B^{2r}\Pi_{P_p}\,\rho\bigr)\,.
\end{equation}
Moves preserve orbits and $I$ is constant on orbits, so $\hat B$ commutes with $\Pi_{P_p}$.

Verification is correct and leaves banknotes undisturbed.  Each move permutes $P_p$, so $\ket{\$_p}$ is fixed by every $\hat P_s$ and hence by $\hat B$.  The banknote is supported on $P_p$, so it is also fixed by $\Pi_{P_p}$.  No step of $\Ver$ changes it, and $\mathrm{Acc}_p(\$_p) = 1$.

\begin{definition}[Measured mint output]\label{def:measured-mint}
The \emph{measured mint output} $\mathcal{M}_n$ is the distribution of pairs $(p, x)$ obtained by running $\Mint(1^n)$ to produce $(p, \ket{\$_p})$ and measuring the banknote in the computational basis.  Equivalently, $x$ is uniform on $X$ and $p = I(x)$.
\end{definition}

\begin{definition}[Path-finding]\label{def:pathfinding-generic}
Let $(X, S, I)$ be an invariant money scheme.  \emph{Path-finding} is the search problem: given $x, y \in X$, output a word $w$ with $w \cdot x = y$, a \emph{path} from $x$ to $y$, or $\bot$ if none exists.
\end{definition}

Farhi et al.\ base the security of the knot scheme on hardness of path-finding~\cite[\S1]{FGHLS10}.  Liu, Montgomery and Zhandry assume path-finding hard for $x$ chosen by the adversary and $y$ uniform on the orbit of $x$~\cite[Asm.~1]{LMZ22}.  \Cref{sec:where-not} relates hardness of path-finding to security.

\paragraph{Instances.}
Two schemes in the literature fit \Cref{def:scheme}; we give the intuition and refer to the original papers for the precise choices.  In the knot scheme of Farhi et al., the objects are grid diagrams, which represent knots, the moves are grid moves, and the invariant is the Alexander polynomial~\cite[\S3]{FGHLS10}.  To keep the set of diagrams finite, they bound the diagrams' dimension and weight each diagram according to its dimension, and with the ground set chosen to reflect that weighting the scheme fits \Cref{def:scheme}.  One consequence matters below: with the dimension bounded, two diagrams of the same knot need not be connected by moves, so one serial number can cover several orbits~\cite[\S3.2.1]{FGHLS10}.  In the scheme of Montgomery and Sharif, the objects are elliptic curves over a prime field, the moves are isogenies of small prime degree, and the invariant is the number of points~\cite[\S3.3, \S6]{MS24}.  There an ideal class group acts on each orbit, and each move acts as one of its elements.

\subsection{The security model}\label{sec:security-model}

Minting and verification are public, so a forger can always produce \emph{some} banknote.  What it must not produce is two states that pass verification at one serial number.  As Liu, Montgomery and Zhandry do~\cite[\S5.5]{LMZ22}, we consider mini-schemes, in which the adversary is given a single banknote~\cite{AC12}.  \Cref{def:security,def:lightning} restate their Definitions~5 and~6.  The two games differ in who produces the banknote.  In quantum lightning the adversary also chooses the serial number, so two copies must be impossible to produce even from a state the adversary constructs itself.

\begin{definition}[Quantum money unforgeability]\label{def:security}
The game between a challenger and an adversary $\calA$ runs as follows.
\begin{enumerate}[label=\textup{\arabic*.},leftmargin=*,nosep]
\item The challenger runs $\Mint(1^n)$, obtaining $(p, \ket{\$_p})$, and gives $p$ and the banknote, held in the \emph{challenge register}, to $\calA$.
\item $\calA$ outputs a possibly entangled state $\rho_{12}$ on two registers $\mathsf R_1,\mathsf R_2$, with reduced states $\rho_1$ and $\rho_2$.
\item The challenger runs $\Ver(p,\rho_1)$ and $\Ver(p,\rho_2)$.  $\calA$ wins if both accept.
\end{enumerate}
$\win_{\calA}(n)$ is the probability that $\calA$ wins.  The scheme is \emph{secure} if $\win_{\calA}(n) = \negl(n)$ for every quantum polynomial-time $\calA$.
\end{definition}

\begin{definition}[Quantum lightning unforgeability]\label{def:lightning}
The game runs as follows.
\begin{enumerate}[label=\textup{\arabic*.},leftmargin=*,nosep]
\item On input $1^n$, the adversary $\calA$ outputs a serial number $p$ and a possibly entangled state $\rho_{12}$ on two registers $\mathsf R_1,\mathsf R_2$, with reduced states $\rho_1$ and $\rho_2$.
\item The challenger runs $\Ver(p,\rho_1)$ and $\Ver(p,\rho_2)$.  $\calA$ wins if both accept.
\end{enumerate}
$\win^{\mathrm L}_{\calA}(n)$ is the probability that $\calA$ wins.  The scheme is \emph{quantum lightning} if $\win^{\mathrm L}_{\calA}(n) = \negl(n)$ for every quantum polynomial-time $\calA$.
\end{definition}

\section{What verification certifies, and design conditions}\label{sec:verification}

This section shows that the states that $\Ver$ accepts with certainty are spanned by the uniform superpositions over single orbits, as Liu, Montgomery and Zhandry observe~\cite[\S8]{LMZ22}, and makes this characterization quantitative in the spectral gap of $\hat B$ and the round count (\Cref{sec:spectral-analysis}).  \Cref{sec:design-conditions} then states three conditions on the scheme and shows what fails without each.  \Cref{sec:osp-problem} then defines the problem that a forger faces after measuring its banknote.

\subsection{Orbit states and accepted weight}\label{sec:spectral-analysis}

We first identify the states that $\hat B$ leaves fixed.  Moves preserve orbits, so $\hat B$ is block-diagonal, with one block $\hat B_k$ for each orbit $\calC_k$.  Each block has four properties.
\begin{itemize}[leftmargin=*,nosep]
\item It is real symmetric, because $\hat B$ is (\Cref{sec:framework}).
\item It is doubly stochastic, being an average of permutation matrices.
\item Its spectrum lies in $[0,1]$.  Laziness gives $\hat B_k = \tfrac12(\openone + M_k)$, where $M_k$ is an average of permutation matrices and so has $\|M_k\| \leq 1$ by the triangle inequality.
\item It is \emph{irreducible}: for all $i,j$ some power of $\hat B_k$ has a positive $(i,j)$ entry, because its orbit is connected in the move graph.
\end{itemize}
By the Perron--Frobenius theorem~\cite{HJ12}, the spectral radius of an irreducible non-negative matrix is a simple eigenvalue.

\begin{lemma}[Orbit states as fixed points of the walk]\label{lem:eigenspace}
For each orbit $\calC_k$, the eigenvalue $1$ of $\hat B_k$ is simple, and its eigenvector is the \emph{orbit state}
\begin{equation}
    \ket{\psi_k} \;:=\; \frac{1}{\sqrt{|\calC_k|}} \sum_{x \in \calC_k} \ket{x}\,.
\end{equation}
\end{lemma}
\begin{proof}
Every row of $\hat B_k$ sums to $1$, so $\hat B_k\ket{\psi_k} = \ket{\psi_k}$, and the spectrum of $\hat B_k$ lies in $[0,1]$, so its spectral radius is $1$.  $\hat B_k$ is non-negative and irreducible, so by the Perron--Frobenius theorem the eigenvalue $1$ is simple.
\end{proof}

Each block $\hat B_k$ has an orthonormal basis of real eigenvectors $\ket{\psi_{k,0}},\dots,\ket{\psi_{k,|\calC_k|-1}}$ with eigenvalues $\lambda_{k,0},\dots,\lambda_{k,|\calC_k|-1}$ in $[0,1]$, indexed so that $\ket{\psi_{k,0}} = \ket{\psi_k}$ and $\lambda_{k,0} = 1$.  By \Cref{lem:eigenspace}, $\lambda_{k,j} < 1$ for every $j \geq 1$.  Set
\[
    \Delta_k \;:=\; 1 - \max_{j \geq 1}\lambda_{k,j}\,, \qquad \Delta \;:=\; \min_k \Delta_k\,,
\]
with $\Delta_k := 1$ when $|\calC_k| = 1$, so that $0 \leq \lambda_{k,j} \leq 1 - \Delta$ for all $k$ and all $j \geq 1$.  We call $\Delta$ the \emph{spectral gap} of the scheme.
For a serial number $p$ let
\[
    \Pi_1^{(p)} \;:=\; \sum_{k :\, I(\calC_k) = p} \ket{\psi_k}\!\bra{\psi_k}
\]
be the projector onto the span of the orbit states with serial number $p$.  The banknote lies in this span: grouping the sum in equation~\eqref{eq:money-generic} by orbits gives $\ket{\$_p} = \sum_{k} (|\calC_k|/|P_p|)^{1/2}\ket{\psi_k}$, the sum running over the orbits with $I(\calC_k) = p$.

\begin{definition}[Accepted weight, accepting state]\label{def:accepted-weight}
The \emph{accepted weight} of a state $\rho$ at serial number $p$ is $\aw_p(\rho) := \Tr\bigl(\Pi_1^{(p)}\rho\bigr)$.  A state of accepted weight $1$ is an \emph{accepting state} at $p$.  The banknote $\ket{\$_p}$ and each orbit state $\ket{\psi_k}$ with $I(\calC_k) = p$ are accepting states at $p$.
\end{definition}

Since $P_p$ is a union of orbits, the vectors $\ket{\psi_{k,j}}$ with $I(\calC_k) = p$ form an orthonormal basis of $\mathrm{span}\{\ket x : x \in P_p\}$.  Every bound below is proved by writing
\begin{equation}\label{eq:eigen-coefficients}
    \Pi_{P_p}\ket{\phi} \;=\; \sum_{k :\, I(\calC_k) = p}\ \sum_{j\geq 0} \hat\phi_{k,j}\ket{\psi_{k,j}}\,, \qquad \hat\phi_{k,j} := \braket{\psi_{k,j}}{\phi}\,,
\end{equation}
so that $\Pi_1^{(p)}\ket\phi = \sum_k \hat\phi_{k,0}\ket{\psi_k}$ and $\aw_p(\phi) = \sum_k |\hat\phi_{k,0}|^2$.

Later arguments replace the acceptance probability by the accepted weight, the expectation of a projector that does not depend on $r$.  The next two bounds justify that replacement.

\begin{lemma}[Verification as an approximate projection]\label{lem:proj-implementation}
Let $p$ be a serial number, $\ket\phi$ a state, and $r \geq 1$.  The unnormalized state left when $\Ver(p,\phi)$ accepts is $\hat B^r\Pi_{P_p}\ket\phi$, and
\[
    \bigl\|\,\hat B^{r}\Pi_{P_p}\ket{\phi} \;-\; \Pi_1^{(p)}\ket{\phi}\,\bigr\| \;\leq\; (1-\Delta)^{r}\,.
\]
\end{lemma}
\begin{proof}
The accepting branch has Kraus operator $\hat B^r\Pi_{P_p}$ (\Cref{sec:framework}).  In the coefficients of equation~\eqref{eq:eigen-coefficients}, $\hat B$ multiplies $\ket{\psi_{k,j}}$ by $\lambda_{k,j}$, and $\lambda_{k,0}=1$, so with all sums over the orbits with $I(\calC_k)=p$,
\begin{align*}
    \hat B^{r}\Pi_{P_p}\ket\phi - \Pi_1^{(p)}\ket\phi &\;=\; \sum_{k}\sum_{j\geq1}\lambda_{k,j}^{r}\,\hat\phi_{k,j}\ket{\psi_{k,j}}\,,\\
    \bigl\|\hat B^{r}\Pi_{P_p}\ket\phi - \Pi_1^{(p)}\ket\phi\bigr\|^2 &\;=\; \sum_{k}\sum_{j\geq1}\lambda_{k,j}^{2r}\,|\hat\phi_{k,j}|^2 \;\leq\; (1-\Delta)^{2r}\,\|\Pi_{P_p}\ket\phi\|^2 \;\leq\; (1-\Delta)^{2r}\,.
\end{align*}
\end{proof}

Verification therefore implements the projective measurement $\{\Pi_1^{(p)},\openone-\Pi_1^{(p)}\}$ up to $(1-\Delta)^r$, which is how Liu, Montgomery and Zhandry specify verification~\cite[\S8.1]{LMZ22} and implement it~\cite[Thm.~11]{LMZ22}.

\begin{lemma}[Acceptance probability versus accepted weight]\label{lem:spectral-overlap}
For every state $\rho$, every serial number $p$ and every $r \geq 1$,
\begin{equation}\label{eq:acc-vs-weight}
    \aw_p(\rho) \;\leq\; \mathrm{Acc}_p(\rho) \;\leq\; \aw_p(\rho) \;+\; \bigl(1-\aw_p(\rho)\bigr)(1-\Delta)^{2r}\,.
\end{equation}
\end{lemma}
\begin{proof}
$\hat B$ commutes with $\Pi_{P_p}$, and in the eigenbasis of equation~\eqref{eq:eigen-coefficients}
\[
    \Pi_{P_p}\hat B^{2r}\Pi_{P_p} \;=\; \Pi_1^{(p)} \;+\; \sum_{k :\, I(\calC_k) = p}\ \sum_{j\geq1} \lambda_{k,j}^{2r}\,\ket{\psi_{k,j}}\!\bra{\psi_{k,j}}\,.
\]
Every $\lambda_{k,j}^{2r}$ with $j \geq 1$ lies in $[0,(1-\Delta)^{2r}]$, so $\Pi_1^{(p)} \preceq \Pi_{P_p}\hat B^{2r}\Pi_{P_p} \preceq \Pi_1^{(p)} + (1-\Delta)^{2r}(\Pi_{P_p} - \Pi_1^{(p)})$.  Taking the trace against $\rho$ and using equation~\eqref{eq:acc-def},
\[
    \aw_p(\rho) \;\leq\; \mathrm{Acc}_p(\rho) \;\leq\; \aw_p(\rho) + (1-\Delta)^{2r}\bigl(\Tr(\Pi_{P_p}\rho) - \aw_p(\rho)\bigr) \;\leq\; \aw_p(\rho) + (1-\Delta)^{2r}\bigl(1-\aw_p(\rho)\bigr)\,.
\]
\end{proof}

\subsection{Design conditions}\label{sec:design-conditions}

Three conditions are imposed on the scheme.  None of them quantifies over an adversary: whether each holds is a combinatorial question about the move graph and the mint, settled scheme by scheme.  The first, rapid mixing (\Cref{def:rapid-mixing}), is what lets verification certify accepted weight.  The other two, dispersed serial numbers (\Cref{def:sparse}) and large orbits (\Cref{def:large-orbits}), are necessary for security (\Cref{prop:design-necessary}), and the three together do not imply security (\Cref{prop:design-insufficient}).

\begin{definition}[Rapid mixing]\label{def:rapid-mixing}
The scheme \emph{mixes rapidly} if $r(n)\,\Delta(n) \geq n$ for all large $n$.
\end{definition}

Since $r$ is polynomial, the condition gives $\Delta \geq n/r \geq 1/\poly(n)$.  Conversely, a scheme with $\Delta \geq 1/q(n)$ for a polynomial $q$ meets it with $r := nq$.  Rapid mixing thus says that the spectral gap is inverse-polynomial and that verification runs enough rounds for it.  Under the condition, $(1-\Delta)^{r} \leq e^{-\Delta r} \leq e^{-n}$, since $1-\Delta \leq e^{-\Delta}$, so by \Cref{lem:spectral-overlap} the acceptance probability and the accepted weight of every state differ by at most $e^{-2n}$.
Without the condition, acceptance does not certify accepted weight: an eigenvector $\ket{\psi_{k,j}}$ with $\lambda_{k,j} = 1-\Delta$ has accepted weight $0$ and is accepted with probability $(1-\Delta)^{2r}$, which tends to $1$ when $\Delta = o(1/r)$.  Such a state passes verification while being orthogonal to every banknote.

\begin{definition}[Dispersed serial numbers]\label{def:sparse}
Let $\mathrm{wt}(p) := |P_p|/|X|$ be the \emph{minting weight} of $p$, the probability that $\Mint$ issues the serial number $p$.  The scheme has \emph{dispersed serial numbers} if $\max_p \mathrm{wt}(p) = \negl(n)$.
\end{definition}

Dispersion is the requirement that honest minting does not forge, which is how Farhi et al.\ first stated it~\cite[\S1.2]{FGHLS10}.

Dispersion does not imply large orbits.  The condition of Farhi et al., that most serial numbers have exponentially many objects~\cite[\S1.2]{FGHLS10}, concerns level sets, and a level set can be a union of many small orbits.  A scheme may have dispersed serial numbers and still mint a non-negligible fraction of its banknotes on orbits of polynomial size, and an orbit of polynomial size can be listed by random walks.

\begin{definition}[Large orbits]\label{def:large-orbits}
The scheme has \emph{large orbits} if for every polynomial $q$,
\[
    \sum_{k:\, |\calC_k| \leq q(n)} \frac{|\calC_k|}{|X|} \;=\; \negl(n)\,.
\]
\end{definition}

\begin{proposition}[Necessity of dispersion and large orbits]\label{prop:design-necessary}
\leavevmode
\begin{enumerate}[label=\textup{(\roman*)},leftmargin=*,nosep]
\item If the scheme does not have dispersed serial numbers, it is not secure.
\item If the scheme mixes rapidly and does not have large orbits, it is not secure.
\end{enumerate}
\end{proposition}
\begin{proof}
(i) A forger keeps the banknote it is given and runs the mint once more.  The new banknote has serial number $p$ with probability $\mathrm{wt}(p)$, so the forger wins with probability
\[
    \sum_p\mathrm{wt}(p)^2 \;\geq\; \max_p\mathrm{wt}(p)^2\,,
\]
which is non-negligible when $\max_p\mathrm{wt}(p)$ is.

(ii) Suppose that for a polynomial $q$ the sum in \Cref{def:large-orbits} is at least $1/\poly(n)$ for infinitely many $n$.  Measuring a banknote yields a uniformly random $x\in X$ (\Cref{def:measured-mint}), so with that probability $x$ lies on an orbit $\calC$ with $|\calC|\le q(n)$.  A walk of $r$ steps from $x$ ends at each $y\in\calC$ with probability
\[
    (\hat B^r)_{yx} \;\geq\; |\calC|^{-1}-(1-\Delta)^r \;\geq\; \frac1{q(n)}-e^{-n} \;\geq\; \frac1{2q(n)}\,,
\]
by the spectral decomposition of $\hat B$ on $\calC$ and \Cref{def:rapid-mixing}.  So $O(q(n)\,n)$ independent walks from $x$ reach every element of $\calC$ except with probability $q(n)(1-1/(2q(n)))^{O(q(n)n)} = 2^{-\Omega(n)}$.  The forger lists the endpoints and prepares two copies of the uniform superposition over the list.  When the list is $\calC$, both copies are $\ket{\psi_{k(x)}}$, which passes verification with certainty (\Cref{lem:eigenspace}), so the forger wins with non-negligible probability.
\end{proof}

Unlike a failure of dispersion, a small gap is not by itself an attack, since the states that pass verification without accepted weight may be hard to prepare~\cite[\S3.3]{FGHLS10},~\cite[App.~F]{LMZ22}.  A small gap does yield an attack when, with non-negligible probability over $\mathcal M_n$, the measured element $x$ lies in a set $J \subseteq \calC_{k(x)}$ of polynomial size that can be listed from $x$ and that the walk rarely leaves.  Farhi et al.\ describe such sets among grid diagrams of maximal dimension~\cite[\S3.3]{FGHLS10}.  Write $\theta(J)$ for the probability that a uniformly random move applied to a uniformly random element of $J$ leaves $J$.  Since $\bra J\hat P_s\ket J = |\{y\in J : s(y)\in J\}|/|J|$ for each move $s$,
\[
    \bra J\hat B\ket J \;=\; \frac12\Bigl(1 + \frac1{|S|}\sum_{s\in S}\bra J\hat P_s\ket J\Bigr) \;=\; \frac12\bigl(1 + 1-\theta(J)\bigr) \;=\; 1-\frac{\theta(J)}2\,.
\]  Since $\ket J$ is supported on $P_p$, expanding it in an orthonormal eigenbasis $\{\ket{e_j}\}$ of $\hat B$, with eigenvalues $\lambda_j\in[0,1]$, and applying Jensen's inequality to the convex function $t \mapsto t^{2r}$ gives
\[
    \mathrm{Acc}_p(\ket J) \;=\; \sum_j |\braket{e_j}{J}|^2\,\lambda_j^{2r} \;\geq\; \Bigl(\sum_j |\braket{e_j}{J}|^2\,\lambda_j\Bigr)^{2r} \;=\; \bigl(1-\theta(J)/2\bigr)^{2r} \;\geq\; 1 - r\,\theta(J)\,.
\]  Two copies of $\ket J$ therefore forge whenever $\theta(J) = o(1/r)$.  Large orbits alone do not exclude this attack, but together with \Cref{def:rapid-mixing} they do: $\ket J$ has accepted weight $|J|/|\calC_{k(x)}|$, so by \Cref{lem:spectral-overlap} it passes with probability at most $|J|/|\calC_{k(x)}| + e^{-2n}$.

\paragraph{The conditions do not imply security.}
The three conditions constrain only the move graph and the mint, and a scheme can meet all of them while its orbit states are easy to prepare.

\begin{proposition}[Insufficiency of the design conditions]\label{prop:design-insufficient}
Let $n$ be even, let $X = \mathbb F_2^n$, let $h_1,\dots,h_{n/2}$ be fixed linearly independent vectors with span $H$, let the moves be $s_i(x) = x + h_i$, let $I(x) := Mx$ for a matrix $M$ with kernel $H$, and let $r := n^2$.  This scheme mixes rapidly and has dispersed serial numbers and large orbits, but it is not secure.
\end{proposition}
\begin{proof}
Each move is its own inverse, so $S = S^{-1}$.  The orbits are the cosets $x + H$, and $I$ separates them.  Every orbit has $2^{n/2}$ elements, and $\mathrm{wt}(p) = 2^{-n/2}$ for every $p$ in the image of $I$.  On each coset the lazy walk is the lazy walk on the hypercube $\{0,1\}^{n/2}$, whose eigenvalues are $1 - 2|z|/n$ for $z\in\{0,1\}^{n/2}$, so its spectral gap is $2/n$ and $r := n^2$ meets \Cref{def:rapid-mixing}.  A forger measures its banknote, obtaining $x$, and prepares the uniform superposition over coefficient vectors $c \in \mathbb F_2^{n/2}$ with $x + \sum_i c_ih_i$ written into a second register.  Since the $h_i$ are linearly independent, $c$ is determined by $x$ and $x+\sum_ic_ih_i$, and solving the linear system erases it:
\[
    2^{-n/4}\sum_{c}\ket{c}\,\Bigl|x+\textstyle\sum_i c_ih_i\Bigr\rangle \;\longmapsto\; \ket{0}\otimes 2^{-n/4}\sum_{c}\Bigl|x+\textstyle\sum_i c_ih_i\Bigr\rangle \;=\; \ket{0}\otimes\ket{x+H}\,.
\]
The second register holds the orbit state, and repeating the construction gives a second copy.  Both copies pass verification with certainty, so the forger wins with probability~$1$.
\end{proof}

The graph scheme of \Cref{sec:coset} also meets the three conditions (\Cref{prop:graph-conditions}) and is broken, although its invariant is known to be computable in polynomial time only on average.

\subsection{The forger's problem: orbit state preparation}\label{sec:osp-problem}

By \Cref{lem:spectral-overlap}, a state that verification accepts at $p$ with probability close to $1$ has accepted weight close to $1$, so it lies close to the span of the orbit states with serial number $p$.  A forger that measures its banknote holds $p$ and one element $x$, which we call its \emph{seed}.  When $I$ takes different values on different orbits, the state it must then produce is the orbit state of $x$.  Producing it is a quantum state generation problem, with a classical input and a quantum state as output, of the kind that Ambainis, Magnin, Roetteler and Roland study in query complexity~\cite{AMRR11}.

\begin{definition}[Orbit state preparation]\label{def:superposition}
Let $G=(G_n)$ with $G_n\subseteq X_n$ be a family of seeds. A uniform family of polynomial-size circuits $\Prep=(\Prep_n)$ on an input register and $m=m(n)$ ancilla qubits \emph{solves orbit state preparation on $G$} if there is a negligible $\mu$ such that for every $n$ and every $x\in G_n$
\begin{equation}\label{eq:osp-condition}
\bigl\|\Prep\ket x\ket{0^m} - \ket{\psi_{k(x)}}\otimes\ket{a_x}\bigr\|\le\mu(n)
\end{equation}
for some ancilla state $\ket{a_x}$.  Such a family is a \emph{solver} for orbit state preparation on $G$.  More generally, a family of states $\ket{\phi_x}$ \emph{can be prepared from $x$} if equation~\eqref{eq:osp-condition} holds with $\ket{\phi_x}$ in place of $\ket{\psi_{k(x)}}$.  The order of the registers in equation~\eqref{eq:osp-condition} is a convention: a circuit that leaves the orbit state in another fixed register becomes a solver after a final swap.  Orbit state preparation is \emph{solvable} if it is solvable on $G=X$. The \emph{minting weight} of $G$ is $|G_n|/|X_n|$, the probability that measuring a freshly minted banknote yields an element of $G_n$.
\end{definition}

The ancilla state may depend on $x$, but the orbit state must come out unentangled from it, so that the element register holds a pure state.

Farhi et al.\ state informally that, given one element, it should be hard to generate the uniform superposition over its level set~\cite[\S1.2]{FGHLS10}.  In Lutomirski's black-box abstraction of the knot scheme, orbit state preparation is the \textsc{Component Superposition} problem~\cite[Def.~7]{Lutomirski11}.

\section{Path-finding and forgery}\label{sec:where-not}

Hardness of path-finding (\Cref{def:pathfinding-generic}) is the assumption behind the knot scheme and the invariant framework (\Cref{sec:framework}).  A path-finder answers a question about two named objects, whereas a forger must produce a state spread coherently over an orbit.  This section asks when the first ability gives the second.  For the graph scheme it does (\Cref{sec:coset}), because of a group structure, which \Cref{sec:group} states as condition~\condH.  When condition~\condH{} holds, hardness of path-finding is necessary for security (\Cref{sec:group}).  The class-group scheme, for which condition~\condH{} is not known to hold, can still be forged in subexponential time (\Cref{sec:group}).  When condition~\condH{} does not hold, a path-finder only turns orbit state preparation into an equally hard problem (\Cref{sec:oracle}).

\subsection{The graph example}\label{sec:coset}

Farhi et al.\ propose a scheme built on graphs, with the spectrum of the adjacency matrix as its invariant, and break it~\cite[\S3.4]{FGHLS10}.  We give the attack in full for a variant whose invariant is a canonical form.

\begin{definition}[The graph scheme]\label{def:graph-scheme}
Let $v\ge2$, and write $\Sym(v)$ for the group of permutations of $\{1,\dots,v\}$.  For $\pi\in\Sym(v)$ and a graph $x$ on the vertices $\{1,\dots,v\}$, let $\pi\cdot x := P_\pi xP_\pi^{-1}$, for the permutation matrix $P_\pi$, be the graph obtained from $x$ by renaming each vertex $i$ as $\pi(i)$.  This is an action: $(\sigma\tau)\cdot x = \sigma\cdot(\tau\cdot x)$.  Write $\mathrm{Aut}(x) := \{\pi\in\Sym(v) : \pi\cdot x = x\}$.  A \emph{canonical labelling} is a map $c$ from graphs to $\Sym(v)$ such that $c(x)\cdot x = c(y)\cdot y$ exactly when $x$ and $y$ are isomorphic.  For a fixed canonical labelling $c$, the \emph{graph scheme} has
\begin{itemize}[leftmargin=*,nosep]
\item \emph{objects}: the adjacency matrices of graphs on $\{1,\dots,v\}$, strings of $n = v(v-1)/2$ bits;
\item \emph{moves}: the renamings $x\mapsto\tau\cdot x$ by transpositions $\tau$;
\item \emph{invariant}: the \emph{canonical form} $I(x) := c(x)\cdot x$.
\end{itemize}
\end{definition}

A function is negligible in $n$ exactly when it is negligible in $v$.  Babai's algorithm computes a canonical labelling in quasipolynomial time~\cite{Babai19}, and on uniformly random graphs the algorithm of Babai and Ku\v{c}era computes one in linear average time~\cite{BK79}.  So requirement~(iii) of \Cref{def:scheme} holds on average and is not known to hold in the worst case.  The transpositions generate $\Sym(v)$, so the orbit of $x$ is its isomorphism class $\{\pi\cdot x : \pi\in\Sym(v)\}$, and since the canonical form separates orbits, each banknote is an orbit state.  The scheme meets the three design conditions of \Cref{sec:design-conditions}, with spectral gap at least $1/(v-1)$ (\Cref{prop:graph-conditions}).

\paragraph{The fact the attack needs.}
Fix $x$ with orbit $\calC$, let $y\in\calC$, and suppose $\sigma\cdot x = y$.  For every $\pi\in\Sym(v)$,
\begin{align*}
    \pi\cdot x = y
    &\iff \pi\cdot x = \sigma\cdot x && \text{since } y = \sigma\cdot x\\
    &\iff (\sigma^{-1}\pi)\cdot x = x && \text{renaming both sides by } \sigma^{-1}\\
    &\iff \sigma^{-1}\pi\in\mathrm{Aut}(x) && \text{by the definition of } \mathrm{Aut}(x)\,.
\end{align*}
Hence
\begin{equation}\label{eq:fibre-coset}
    \{\pi\in\Sym(v) : \pi\cdot x = y\} \;=\; \sigma\,\mathrm{Aut}(x) \;:=\; \{\sigma\alpha : \alpha\in\mathrm{Aut}(x)\}\,.
\end{equation}  The canonical labelling supplies such a $\sigma$.  Since $x$ and $y$ are isomorphic, $c(x)\cdot x = I(x) = I(y) = c(y)\cdot y$, so
\begin{equation}\label{eq:canonical-section}
    \sigma_y := c(y)^{-1}c(x) \qquad\text{satisfies}\qquad \sigma_y\cdot x \;=\; c(y)^{-1}\cdot\bigl(c(x)\cdot x\bigr) \;=\; c(y)^{-1}\cdot\bigl(c(y)\cdot y\bigr) \;=\; y\,.
\end{equation}
The $|\calC|$ sets $\sigma_y\mathrm{Aut}(x)$, the \emph{left cosets} of $\mathrm{Aut}(x)$, partition $\Sym(v)$, so
\begin{equation}\label{eq:orbit-stabilizer}
    |\calC| \;=\; v!/|\mathrm{Aut}(x)|\,.
\end{equation}

\paragraph{The attack.}
The forger measures its banknote, obtaining $x$, and computes $c(x)$.  It holds a permutation register and a graph register.
\begin{enumerate}[label=\textup{Step \arabic*.},leftmargin=*,itemsep=4pt]
\item Prepare
\[
    \ket{\Sym(v)}\ket{x} \;=\; \frac1{\sqrt{v!}}\sum_{\pi\in\Sym(v)}\ket{\pi}\ket{x}\,.
\]
A permutation can be built by inserting the numbers $1,2,\dots,v$ one at a time into a list, and it is recorded by the positions $d_1,\dots,d_v$ at which they are inserted.  When the number $i$ is inserted, the list holds the $i-1$ numbers already placed, so there are $i$ positions and $d_i\in\{0,\dots,i-1\}$; numbers placed early move to higher positions as later ones are inserted before them.  Every arrangement of the $v$ numbers arises from exactly one digit string, and there are $1\cdot2\cdots v = v!$ strings.  So $\ket{\Sym(v)}$ is obtained from the product state
\[
    \bigotimes_{i=1}^{v}\;\frac1{\sqrt i}\sum_{d=0}^{i-1}\ket{d}
\]
by a reversible conversion of digits into a permutation.  Each factor, a uniform superposition over at most $v$ values, can be prepared up to negligible error by a circuit of size $\poly(v)$.
\item Apply $\ket{\pi}\ket{x}\mapsto\ket{\pi}\ket{\pi\cdot x}$ and group the terms by $y = \pi\cdot x$, using equation~\eqref{eq:fibre-coset}:
\[
    \frac1{\sqrt{v!}}\sum_{\pi}\ket\pi\ket{\pi\cdot x} \;=\; \sum_{y\in\calC}\sqrt{\frac{|\mathrm{Aut}(x)|}{v!}}\;\ket{\sigma_y\mathrm{Aut}(x)}\ket{y}\,.
\]
Each $y$ is entangled with the coset recorded in the first register.
\item Controlled on $y$, compute $\sigma_y$ by equation~\eqref{eq:canonical-section}, replace $\pi$ by $\sigma_y^{-1}\pi$, and uncompute $\sigma_y$.  Every coset $\sigma_y\mathrm{Aut}(x)$ becomes $\mathrm{Aut}(x)$, and $\sqrt{|\mathrm{Aut}(x)|/v!} = |\calC|^{-1/2}$ by equation~\eqref{eq:orbit-stabilizer}, so the state becomes
\[
    \sum_{y\in\calC}|\calC|^{-1/2}\,\ket{\mathrm{Aut}(x)}\ket{y} \;=\; \ket{\mathrm{Aut}(x)}\otimes\ket{\psi_{k(x)}}\,.
\]
\end{enumerate}
Repeating the three steps gives a second copy.  The forger's cost is that of computing $c$ coherently in step~3, which with Babai's algorithm is quasipolynomial.  With the average-case algorithm of Babai and Ku\v{c}era~\cite{BK79}, stopped after a fixed polynomial time, the forger runs in polynomial time and errs only on the branches $y$ where that algorithm has not finished.  For a uniformly random seed $x$, the branch $y$ is itself a uniformly random graph, so by Markov's inequality these branches carry inverse-polynomially small weight on average, and the forger wins with probability close to $1$.  Step~3 computed $\sigma_y$, a product of at most $v-1$ transpositions and hence a path from $x$ to $y$.  Canonical labelling is thus a path-finder for this scheme, run coherently.

\subsection{Index sets}\label{sec:index-sets}

The attack used three ingredients.  We call the elements of the set in the first one \emph{labels}: a label names an object of the orbit, as the permutation $\pi$ names the graph $\pi\cdot x$.
\begin{itemize}[leftmargin=*,nosep]
\item a set $W$ of labels whose uniform superposition is easy to prepare, here $\Sym(v)$;
\item a map $\varphi_x$ sending each label to the object it names, onto the orbit, here $\pi\mapsto\pi\cdot x$;
\item for each $y$ in the orbit, a \emph{relabelling} $\tau_y$ of $W$, applied in the branch of $y$, here $\pi\mapsto\sigma_y^{-1}\pi$.
\end{itemize}
Steps~1 and~2 need only the first two ingredients.  The relabelling was used in step~3, to obtain $\ket{\mathrm{Aut}(x)}$ for every $y$.  \Cref{lem:uncompute} shows that this is exactly what step~3 needs, for any set $W$.

\begin{lemma}[Label erasure by relabelling]\label{lem:uncompute}
Fix a seed $x$ and write $\calC:=\calC_{k(x)}$.  Let $W$ be a finite set, let $\varphi_x:W\to X$ be a map whose image is $\calC$, write $W_{x\to y} := \{w\in W:\varphi_x(w) = y\}$ for the \emph{fibre} over $y\in\calC$, and let
\[
    \ket\Psi:=|W|^{-1/2}\sum_{w\in W}\ket w_{\mathsf W}\ket{\varphi_x(w)}_{\mathsf X}\,.
\]  For each $y\in\calC$ let $\tau_y$ be a \emph{relabelling} of $W$, that is, a bijection $\tau_y:W\to W$, and let $\Gamma$ be the unitary that applies $\tau_y$ to $\mathsf W$ when $\mathsf X$ holds $y\in\calC$, and acts as the identity when $\mathsf X$ holds an element outside $\calC$.  Then
\[
\Gamma\ket\Psi = \sum_{y\in\calC}\sqrt{|W_{x\to y}|/|W|}\;\ket{\tau_y(W_{x\to y})}\ket y,
\]
and $\Gamma\ket\Psi$ is a product state across $\mathsf W$ and $\mathsf X$ if and only if the relabellings satisfy the \emph{label condition}
\begin{equation}\label{eq:label-condition}
    \tau_y(W_{x\to y}) \;=\; \tau_{y'}(W_{x\to y'}) \qquad\text{for all } y,y'\in\calC\,.
\end{equation}
When the label condition holds, the fibres have a common size and $\mathsf X$ holds $\ket{\psi_{k(x)}}$ exactly.
\end{lemma}
\begin{proof}
Group the terms of $\ket\Psi$ by the value $y=\varphi_x(w)$.  The labels with that value form the fibre $W_{x\to y}$, so
\[
    \ket\Psi \;=\; \sum_{y\in\calC}\;\sum_{w\in W_{x\to y}}|W|^{-1/2}\,\ket w\ket y \;=\; \sum_{y\in\calC}\sqrt{\frac{|W_{x\to y}|}{|W|}}\;\ket{W_{x\to y}}\ket y\,.
\]
In the branch of $y$, $\Gamma$ applies $\tau_y$, which maps the uniform superposition over $W_{x\to y}$ to the uniform superposition over its image:
\[
    \Gamma\ket\Psi \;=\; \sum_{y\in\calC}\sqrt{\frac{|W_{x\to y}|}{|W|}}\;\ket{\tau_y(W_{x\to y})}\ket y\,.
\]
A state $\sum_y c_y\ket{a_y}\ket y$ with orthonormal $\ket y$, coefficients $c_y>0$ and unit vectors $\ket{a_y}$ is a product state exactly when all the $\ket{a_y}$ are equal up to a phase.  Here $\ket{a_y} = \ket{\tau_y(W_{x\to y})}$ has non-negative real entries, so equality up to a phase is equality.  Two uniform superpositions $\ket A$ and $\ket{A'}$ over subsets of $W$ are equal exactly when $A = A'$, which is the label condition.  If it holds with common image $A$, then, since each $\tau_y$ is a bijection,
\[
    |W_{x\to y}| \;=\; |\tau_y(W_{x\to y})| \;=\; |A| \quad\text{for all } y\,, \qquad\text{so}\qquad |W_{x\to y}| = |W|/|\calC|\,,
\]
the last because the $|\calC|$ fibres partition $W$.  Every coefficient is therefore $|\calC|^{-1/2}$, and
\[
    \Gamma\ket\Psi \;=\; \ket A\otimes|\calC|^{-1/2}\sum_{y\in\calC}\ket y \;=\; \ket A\otimes\ket{\psi_{k(x)}}\,.
\]
\end{proof}

In the graph example the fibre over $y$ is $\sigma_y\mathrm{Aut}(x)$, and $\tau_y$ is multiplication on the left by $\sigma_y^{-1}$, which sends it to $\mathrm{Aut}(x)$ for every $y$, so the label condition holds.

\begin{definition}[Index set, section, coset parameterization]\label{def:index-set}
Fix a seed $x$ and write $\calC:=\calC_{k(x)}$.
\begin{enumerate}[label=\textup{(\roman*)},leftmargin=*,nosep]
\item An \emph{index set for $\calC$ at $x$} is a set $W$ of bit strings of polynomial length, the \emph{labels}, with a polynomial-time map $\varphi_x:W\to X$ such that $\varphi_x(W)=\calC$.  The label $w$ names the object $\varphi_x(w)$.
\item A \emph{section} of $(W,\varphi_x)$ is a map $u:\calC\to W$ with $\varphi_x(u(y))=y$ for all $y\in\calC$.
\item $(W,\varphi_x)$ is a \emph{coset parameterization} if $W$ is a finite group, with identity $1_W$, that acts on $\calC$ with $\varphi_x(w)=w\cdot x$ for all $w\in W$, and if the following are computable in polynomial time: multiplication and inversion in $W$, the preparation of $\ket W$ to within negligible error, and the action of $W$ on $\calC$.
\item A coset parameterization \emph{acts through the moves} if for each $s\in S$ there is $g_s\in W$, computable in polynomial time from $s$, with $\varphi_x(g_sw)=s(\varphi_x(w))$ for all $w\in W$.
\end{enumerate}
\end{definition}

A section chooses one label in each fibre.  One always exists.  The results below depend on whether one is computable in polynomial time.

\begin{proposition}[Orbit state preparation from a section]\label{prop:index-set}
Let $G$ be a family of seeds, and for each $x\in G_n$ let $(W_x,\varphi_x)$ be an index set for $\calC_{k(x)}$ at $x$ with a section $u_x$, where $\varphi_x(w)$ and $u_x(y)$ are computable in time $\poly(n)$ from $(x,w)$ and $(x,y)$.  Write
\[
    \mathrm{im}\,u_x \;:=\; \{u_x(y) : y\in\calC_{k(x)}\}\;\subseteq\;W_x
\]
for the image of the section.  Then orbit state preparation is solvable on $G$ if and only if $\ket{\mathrm{im}\,u_x}$ can be prepared from $x$, for every $x\in G_n$.
\end{proposition}
\begin{proof}
Fix $x \in G_n$, and write $u = u_x$ and $\calC := \calC_{k(x)}$.  In the chains below, ``undo'' runs a computation backwards, clearing the register it wrote.
The map $\varphi_x$ restricts to a bijection from $\mathrm{im}\,u$ onto $\calC$, with inverse $u$.  It is onto because $\varphi_x(u(y)) = y$, and one-to-one because $\varphi_x(u(y)) = \varphi_x(u(y'))$ gives $y = y'$.  In particular $|\mathrm{im}\,u| = |\calC|$, and $u(\varphi_x(w)) = w$ for every $w\in\mathrm{im}\,u$.  With two fresh registers, the orbit state is prepared from $\ket{\mathrm{im}\,u}$ by
\begin{align*}
    |\calC|^{-1/2}\sum_{w\in\mathrm{im}\,u}\ket w\ket0\ket0
    &\;\overset{\varphi_x}{\longmapsto}\; |\calC|^{-1/2}\sum_{w\in\mathrm{im}\,u}\ket w\ket{\varphi_x(w)}\ket0\\
    &\;\overset{u}{\longmapsto}\; |\calC|^{-1/2}\sum_{w\in\mathrm{im}\,u}\ket w\ket{\varphi_x(w)}\ket w\\
    &\;\overset{\oplus}{\longmapsto}\; |\calC|^{-1/2}\sum_{w\in\mathrm{im}\,u}\ket 0\ket{\varphi_x(w)}\ket w\\
    &\;\overset{\text{undo }u}{\longmapsto}\; |\calC|^{-1/2}\sum_{w\in\mathrm{im}\,u}\ket 0\ket{\varphi_x(w)}\ket 0
    \;=\; \ket0\ket{\psi_{k(x)}}\ket0\,.
\end{align*}
The step labelled $u$ writes $u(\varphi_x(w)) = w$, the step labelled $\oplus$ adds the third register into the first, and the last equality holds because $w\mapsto\varphi_x(w)$ is a bijection from $\mathrm{im}\,u$ onto $\calC$.  Conversely, the orbit state gives $\ket{\mathrm{im}\,u}$:
\begin{align*}
    |\calC|^{-1/2}\sum_{y\in\calC}\ket y\ket0\ket0
    &\;\overset{u}{\longmapsto}\; |\calC|^{-1/2}\sum_{y\in\calC}\ket y\ket{u(y)}\ket0\\
    &\;\overset{\varphi_x}{\longmapsto}\; |\calC|^{-1/2}\sum_{y\in\calC}\ket y\ket{u(y)}\ket y\\
    &\;\overset{\oplus}{\longmapsto}\; |\calC|^{-1/2}\sum_{y\in\calC}\ket 0\ket{u(y)}\ket y\\
    &\;\overset{\text{undo }\varphi_x}{\longmapsto}\; |\calC|^{-1/2}\sum_{y\in\calC}\ket 0\ket{u(y)}\ket 0
    \;=\; \ket0\ket{\mathrm{im}\,u}\ket0\,,
\end{align*}
since $y\mapsto u(y)$ is a bijection from $\calC$ onto $\mathrm{im}\,u$.  Each chain is a fixed unitary acting on the output register and two fresh registers.  So if a circuit prepares one of the two states to within $\mu$, with ancilla state $\ket{a_x}$, then applying the chain prepares the other to within $\mu$, with ancilla state $\ket{0}\ket0\ket{a_x}$.
\end{proof}

The proposition needs no relabellings, because $\varphi_x$ is injective on $\mathrm{im}\,u_x$, but it requires preparing $\ket{\mathrm{im}\,u_x}$.

\subsection{When condition~\texorpdfstring{\condH}{(H)} holds}\label{sec:group}

\begin{definition}[Condition~\condH]\label{def:condH}
A family of seeds $G$ satisfies \emph{condition~\condH} if, for each $x\in G_n$, a coset parameterization $(W_x,\varphi_x)$ of $\calC_{k(x)}$ at $x$ that acts through the moves is computable in polynomial time from $x$.  A scheme satisfies condition~\condH{} if the family of all seeds does.
\end{definition}

The graph scheme satisfies condition~\condH: $W = \Sym(v)$ with $\varphi_x(\pi) = \pi\cdot x$ is a coset parameterization acting through the moves, with $g_s$ the transposition that the move $s$ renames by.  Canonical labelling supplies the section $u(y) = \sigma_y = c(y)^{-1}c(x)$ of step~3.  The next lemma records the two facts about coset parameterizations that this section uses.

\begin{lemma}[Stabilizer cosets and path products]\label{lem:coset-fibres}
Let $(W,\varphi_x)$ be a coset parameterization of $\calC := \calC_{k(x)}$, and write $\mathrm{Stab}(x) := W_{x\to x}$.
\begin{enumerate}[label=\textup{(\roman*)},leftmargin=*,nosep]
\item $\mathrm{Stab}(x)$ is a subgroup of $W$, and the fibre over each $y\in\calC$ is the left coset $w\,\mathrm{Stab}(x)$ for any $w$ in it.  In particular every fibre has $|\mathrm{Stab}(x)|$ elements, and $|W| = |\calC|\,|\mathrm{Stab}(x)|$.
\item If $(W,\varphi_x)$ acts through the moves and $s_\ell\cdots s_1$ is a path from $x$ to $y$, then $g_{s_\ell}\cdots g_{s_1}\in W_{x\to y}$.  So a path-finder that succeeds on every target supplies a section.
\end{enumerate}
\end{lemma}
\begin{proof}
(i) $\mathrm{Stab}(x)$ is the stabilizer of $x$ under the action, hence a subgroup.  For $w,w'\in W$,
\begin{align*}
    \varphi_x(w) = \varphi_x(w')
    &\iff w\cdot x = w'\cdot x && \text{by the definition of } \varphi_x\\
    &\iff (w^{-1}w')\cdot x = x && \text{acting on both sides by } w^{-1}\\
    &\iff w^{-1}w'\in\mathrm{Stab}(x) && \text{by the definition of } \mathrm{Stab}(x)\,,
\end{align*}
so the fibre containing $w$ is $w\,\mathrm{Stab}(x)$.  The $|\calC|$ fibres partition $W$ and each has $|\mathrm{Stab}(x)|$ elements.

(ii) By induction on $\ell$.  For $\ell=0$, $\varphi_x(1_W) = x$.  If $\varphi_x(g_{s_{\ell-1}}\cdots g_{s_1}) = s_{\ell-1}\cdots s_1(x)$, then by \Cref{def:index-set}(iv),
\[
    \varphi_x(g_{s_\ell}g_{s_{\ell-1}}\cdots g_{s_1}) \;=\; s_\ell\bigl(\varphi_x(g_{s_{\ell-1}}\cdots g_{s_1})\bigr) \;=\; s_\ell s_{\ell-1}\cdots s_1(x)\,,
\]
which is $y$ at the end of a path from $x$ to $y$.  A section is obtained by mapping each $y$ to the product along the path that the path-finder returns.
\end{proof}

In the graph example $\mathrm{Stab}(x) = \mathrm{Aut}(x)$, and part~(i) of the lemma is equation~\eqref{eq:fibre-coset}.

\begin{proposition}[Orbit state preparation from a coset parameterization]\label{prop:coset-osp}
Let $G$ be a family of seeds, and for each $x\in G_n$ let $(W_x,\varphi_x)$ be a coset parameterization of $\calC_{k(x)}$ at $x$ with a section $u_x$ computable in time $\poly(n)$ from $(x,y)$.  Then a circuit maps $\ket{W_x}$ to $\ket{\mathrm{Stab}(x)}\otimes\ket{\psi_{k(x)}}$ exactly, using one evaluation of $\varphi_x$, two evaluations of $u_x$ and $O(1)$ group operations.  In particular, orbit state preparation is solvable on $G$.
\end{proposition}
\begin{proof}
Fix $x \in G_n$, and write $u = u_x$, $\calC := \calC_{k(x)}$ and $\mathcal S := \mathrm{Stab}(x)$.  As in the proof of \Cref{prop:index-set}, ``undo'' runs a computation backwards, clearing the register it wrote.
By \Cref{lem:coset-fibres}(i), the fibre over $y$ is $u(y)\mathcal S$, because $u(y)$ lies in it, and $|W_x| = |\calC|\,|\mathcal S|$.  So every $w\in W_x$ is $w = u(y)\alpha$ for exactly one pair $y\in\calC$, $\alpha\in\mathcal S$.  With two fresh registers,
\begin{align*}
    |W_x|^{-1/2}\sum_{y\in\calC}\sum_{\alpha\in\mathcal S}\ket{u(y)\alpha}\ket0\ket0
    &\;\overset{\varphi_x}{\longmapsto}\; |W_x|^{-1/2}\sum_{y,\alpha}\ket{u(y)\alpha}\ket y\ket0\\
    &\;\overset{u}{\longmapsto}\; |W_x|^{-1/2}\sum_{y,\alpha}\ket{u(y)\alpha}\ket y\ket{u(y)}\\
    &\;\overset{\mathrm{mult}}{\longmapsto}\; |W_x|^{-1/2}\sum_{y,\alpha}\ket{\alpha}\ket y\ket{u(y)}\\
    &\;\overset{\text{undo }u}{\longmapsto}\; |W_x|^{-1/2}\sum_{y,\alpha}\ket{\alpha}\ket y\ket0
    \;=\; \ket{\mathcal S}\otimes\ket{\psi_{k(x)}}\otimes\ket0\,,
\end{align*}
where the step labelled $\mathrm{mult}$ multiplies the first register on the left by the inverse of the third, and the last equality uses $|W_x| = |\calC|\,|\mathcal S|$.  The chain uses one evaluation of $\varphi_x$, two of $u$ and $O(1)$ group operations.  The first step produces the state $\ket\Psi$ of \Cref{lem:uncompute}, and the remaining three are its unitary $\Gamma$, with the relabellings $\tau_y(w) = u(y)^{-1}w$, which send every fibre $u(y)\mathcal S$ to $\mathcal S$.  The ancilla state $\ket{\mathcal S}$ depends on $x$ only, and $\ket{W_x}$ is prepared to within negligible error (\Cref{def:index-set}), so orbit state preparation is solvable on $G$.
\end{proof}

Unlike \Cref{prop:index-set}, this superposes over all of $W_x$, which is easy, and then needs relabellings that send every fibre to $\mathrm{Stab}(x)$.  Each fibre is a left coset of $\mathrm{Stab}(x)$ by \Cref{lem:coset-fibres}(i), and the section supplies the relabellings.

Under condition~\condH, \Cref{lem:coset-fibres}(ii) turns a path-finder into the section that \Cref{prop:coset-osp} needs.  The theorem below asks only for a path-finder that succeeds on a random target, which is the form in which path-finding is assumed hard.  Its proof uses the group to reach every target $y$ from a random one.  For uniformly random $w\in W$, the target $w\cdot y$ is uniform on the orbit.  A path from $x$ to $w\cdot y$ multiplies out to some $h\in W$ with $h\cdot x = w\cdot y$ (\Cref{lem:coset-fibres}(ii)), and then $u := w^{-1}h$ satisfies $u\cdot x = w^{-1}\cdot(w\cdot y) = y$.

\begin{theorem}[Necessity of path-finding hardness under condition~\condH]\label{thm:pf-necessary}
Let $G$ be a family of seeds that satisfies condition~\condH, with coset parameterizations $(W_x,\varphi_x)$.
\begin{enumerate}[label=\textup{(\roman*)},leftmargin=*,nosep]
\item If a quantum polynomial-time algorithm $\calA$ and a polynomial $q$ satisfy, for every $x\in G_n$,
$\Pr_{y}\bigl[\calA(x,y)\text{ outputs a path from }x\text{ to }y\bigr]\ge1/q(n)$, with $y$ uniform on $\calC_{k(x)}$, then orbit state preparation is solvable on $G$.
\item If some quantum polynomial-time $\calA$ and polynomial $q$ satisfy
\[
\Pr\bigl[x\in G_n\text{ and }\calA(x,y)\text{ outputs a path from }x\text{ to }y\bigr]\;\ge\;\frac1{q(n)}
\]
for infinitely many $n$, where $(p,x)\leftarrow\mathcal M_n$ and $y$ is uniform on $\calC_{k(x)}$, then a forger wins with probability at least $1/(3q(n))$ for infinitely many $n$, so the scheme is not secure.
\end{enumerate}
\end{theorem}
\begin{proof}
(i) Fix $x\in G_n$, and write $\calC=\calC_{k(x)}$, $W=W_x$, $\mathcal S=\mathrm{Stab}(x)$ and $N:=\lceil2n\,q(n)\rceil$.  For a word $\omega = s_\ell\cdots s_1$ write $h(\omega) := g_{s_\ell}\cdots g_{s_1}\in W$, and write $\calA(x,z;c)$ for the word that $\calA$ outputs on $(x,z)$ with coins $c\in\{0,1\}^m$.  We first assume that $\ket W$ is prepared exactly, and account for the preparation error at the end.

\emph{A section from the path-finder.}  For $y\in\calC$, $w\in W$ and coins $c$, let
\[
    h_{w,c} \;:=\; h\bigl(\calA(x,\,w\cdot y;\,c)\bigr)\,,\qquad b_{w,c} \;:=\; \bigl[\,h_{w,c}\cdot x = w\cdot y\,\bigr]\in\{0,1\}\,.
\]
If $b_{w,c}=1$, then
\[
    \bigl(w^{-1}h_{w,c}\bigr)\cdot x \;=\; w^{-1}\cdot(h_{w,c}\cdot x) \;=\; w^{-1}\cdot(w\cdot y) \;=\; y\,,\qquad\text{so}\qquad w^{-1}h_{w,c}\in W_{x\to y}\,.
\]
The same group acting at $y$, $w\mapsto w\cdot y$, is a coset parameterization of $\calC$ at $y$, so by \Cref{lem:coset-fibres}(i) its fibres all have $|W|/|\calC|$ elements, and it takes the uniform distribution on $W$ to the uniform distribution on $\calC$.  And by \Cref{lem:coset-fibres}(ii), $b_{w,c}=1$ whenever $\calA(x,w\cdot y;c)$ is a path.  So
\[
    \Pr_{w,c}\bigl[b_{w,c}=1\bigr] \;\ge\; \Pr_{z\leftarrow\calC,\,c}\bigl[\calA(x,z;c)\text{ is a path from }x\text{ to }z\bigr] \;\ge\; \frac1{q(n)}\,.
\]
Let $\Sec$ be the unitary that, controlled on $\ket y$, runs $N$ trials, each preparing $\ket W\otimes2^{-m/2}\sum_c\ket c$ in fresh registers and computing $h_{w,c}$ and $b_{w,c}$ coherently.  It then writes into a register $\mathsf U$ the element $u = w_j^{-1}h_{w_j,c_j}$ for the least $j$ with $b_{w_j,c_j}=1$, or $u=1_W$ if there is none.  Split its output as
\[
    \Sec\ket y\ket0 \;=\; \ket y\otimes\bigl(\ket{\mathrm{good}_y}+\ket{\mathrm{bad}_y}\bigr)\,,
\]
where $\ket{\mathrm{bad}_y}$ is the component on which every $b_{w_j,c_j}$ is $0$.  On $\ket{\mathrm{good}_y}$ the register $\mathsf U$ holds an element of $W_{x\to y}$.  The trials are independent, so
\[
    \|\ket{\mathrm{bad}_y}\|^2 \;=\; \prod_{j=1}^N\Pr_{w_j,c_j}\bigl[b_{w_j,c_j}=0\bigr] \;\le\; \Bigl(1-\frac1{q(n)}\Bigr)^N \;\le\; e^{-N/q(n)} \;\le\; e^{-2n}\,.
\]

\emph{The circuit.}  Prepare $\ket W$ and compute $\varphi_x$ into the element register.  By \Cref{lem:coset-fibres}(i), the fibres have size $|\mathcal S|$ and $|W| = |\calC|\,|\mathcal S|$, so
\[
    |W|^{-1/2}\sum_{w\in W}\ket w\ket{\varphi_x(w)} \;=\; |\calC|^{-1/2}\sum_{y\in\calC}\ket{W_{x\to y}}\ket y\,.
\]
Then apply $\Sec$, the multiplication $M:\ket w\ket u\mapsto\ket{u^{-1}w}\ket u$ with $u$ read from $\mathsf U$, and $\Sec^\dagger$:
\begin{align*}
    |\calC|^{-1/2}\sum_{y}\ket{W_{x\to y}}\ket y\ket0
    &\;\overset{\Sec}{\longmapsto}\; |\calC|^{-1/2}\sum_{y}\ket{W_{x\to y}}\ket y\bigl(\ket{\mathrm{good}_y}+\ket{\mathrm{bad}_y}\bigr)\\
    &\;\overset{M}{\longmapsto}\; |\calC|^{-1/2}\sum_{y}\Bigl(\ket{\mathcal S}\ket y\bigl(\ket{\mathrm{good}_y}+\ket{\mathrm{bad}_y}\bigr) + \ket{e_y}\Bigr)\\
    &\;\overset{\Sec^\dagger}{\longmapsto}\; \ket{\mathcal S}\otimes\ket{\psi_{k(x)}}\otimes\ket0 \;+\; |\calC|^{-1/2}\sum_{y}\Sec^\dagger\ket{e_y}\,.
\end{align*}
The second step uses that, for every $u\in W_{x\to y}$, \Cref{lem:coset-fibres}(i) gives $W_{x\to y} = u\,\mathcal S$, so
\[
    M\bigl(\ket{W_{x\to y}}\ket u\bigr) \;=\; \ket{u^{-1}W_{x\to y}}\ket u \;=\; \ket{\mathcal S}\ket u\,,
\]
where $\ket{e_y} := M\bigl(\ket{W_{x\to y}}\ket y\ket{\mathrm{bad}_y}\bigr) - \ket{\mathcal S}\ket y\ket{\mathrm{bad}_y}$ is the error from the bad branch, with $\|\ket{e_y}\|\le2\|\ket{\mathrm{bad}_y}\|\le2e^{-n}$.  The third step uses $\Sec^\dagger\bigl(\ket y(\ket{\mathrm{good}_y}+\ket{\mathrm{bad}_y})\bigr) = \ket y\ket0$.  Each $\ket{e_y}$ keeps $\ket y$ in the element register, and $\Sec^\dagger$ is controlled on it, so the terms $\Sec^\dagger\ket{e_y}$ are orthogonal for different $y$, and
\[
    \Bigl\||\calC|^{-1/2}\sum_{y}\Sec^\dagger\ket{e_y}\Bigr\| \;=\; \Bigl(|\calC|^{-1}\sum_{y}\|\ket{e_y}\|^2\Bigr)^{1/2} \;\le\; 2e^{-n}\,.
\]

\emph{The preparation error.}  Let $P$ prepare $\ket W$ with $\|P\ket0-\ket W\|\le\mu$, choosing the global phase of $\ket W$ so that $\bra WP\ket0\ge0$ (\Cref{def:index-set}).  Let $R$ be the rotation in the plane of $P\ket0$ and $\ket W$ that takes $P\ket0$ to $\ket W$.  Then $RP$ prepares $\ket W$ exactly, and
\[
    \|RP-P\| \;=\; \|(RP)^\dagger-P^\dagger\| \;=\; \|R-\openone\| \;=\; \|P\ket0-\ket W\| \;\le\; \mu\,.
\]
The circuit uses $P$ or $P^\dagger$ $2N+1$ times: once at the start, once in each trial of $\Sec$, and once in each trial of $\Sec^\dagger$.  Replacing each use by its exact version therefore changes the final state by at most $(2N+1)\mu = \negl(n)$.  So the circuit solves orbit state preparation on $G$, with ancilla state $\ket{\mathcal S}\otimes\ket0$.

(ii) Suppose the probability in part~(ii) is at least $1/q(n)$ for infinitely many $n$.  For $x\in G_n$ let $s(x) := \Pr_y[\calA(x,y)\text{ outputs a path from }x\text{ to }y]$, and let $G'_n := \{x\in G_n : s(x)\ge1/(2q(n))\}$.  With $x$ uniform on $X_n$ (\Cref{def:measured-mint}), $\mathbb E_x[\mathbf 1_{G_n}(x)\,s(x)]\ge1/q(n)$ and $s(x)\le1$, so
\[
    \frac1{q(n)} \;\le\; \Pr_x\bigl[x\in G'_n\bigr]\cdot1+\frac1{2q(n)}\,,\qquad\text{so}\qquad \Pr_x\bigl[x\in G'_n\bigr]\;\ge\;\frac1{2q(n)}\,.
\]
A forger measures its banknote, obtaining $(p,x)\leftarrow\mathcal M_n$, runs the circuit of part~(i) with $2q$ in place of $q$ twice, and outputs the two element registers.  For $x\in G'_n$, each element register is within trace distance $\varepsilon_n := 2e^{-n}+\negl(n)$ of $\ket{\psi_{k(x)}}$.  It therefore has accepted weight at least $1-\varepsilon_n$, and passes with at least that probability (\Cref{lem:spectral-overlap}).  The two runs are independent, so
\[
    \Pr[\text{the forger wins}] \;\ge\; \Pr_x\bigl[x\in G'_n\bigr]\,(1-\varepsilon_n)^2 \;\ge\; \frac1{2q(n)}(1-\varepsilon_n)^2 \;\ge\; \frac1{3q(n)}
\]
for infinitely many $n$.
\end{proof}

For the class-group scheme of Montgomery and Sharif, condition~\condH{} is not known to hold: a group acts on each orbit and every move acts as one of its elements, but the action of an arbitrary group element is not known to be efficient.  The best known algorithms still give a subexponential forgery (\Cref{app:elliptic-curves}).

\begin{proposition}[Subexponential forgery of the class-group scheme]\label{prop:ms-forgery}
Assume the generalized Riemann hypothesis.  There is a quantum algorithm that runs in time $2^{o(n)}$ and, given a banknote $\ket{\$_p}$ of the scheme of Montgomery and Sharif~\cite{MS24}, outputs registers in states $\rho_1,\rho_2$ with
\[
    \Pr\bigl[\Ver(p,\rho_1)\text{ and }\Ver(p,\rho_2)\text{ both accept}\bigr] \;\ge\; 1-2^{-\Omega(n)}\,.
\]
\end{proposition}
\begin{proof}
    The proof is deferred to \Cref{app:elliptic-curves}.
\end{proof}

\subsection{When condition~\texorpdfstring{\condH}{(H)} does not hold}\label{sec:oracle}

For knot orbits no coset parameterization is known, so condition~\condH{} is not known to hold and \Cref{prop:coset-osp} does not apply.  A path-finder still supplies a section, and \Cref{prop:index-set} says what it achieves.

\begin{corollary}[Orbit state preparation from a path-finder]\label{cor:pf-section}
Let $G$ be a family of seeds, and let $\calA$ be a deterministic polynomial-time algorithm that outputs a path $\calA(x,y)$ from $x$ to $y$ for every $x\in G_n$ and $y\in\calC_{k(x)}$.  Write
\[
    \Ima_x \;:=\; \{\calA(x,y) : y\in\calC_{k(x)}\}
\]
for the set of paths it outputs from $x$.  Then orbit state preparation is solvable on $G$ if and only if $\ket{\Ima_x}$ can be prepared from $x$, for every $x\in G_n$.
\end{corollary}
\begin{proof}
Let $\ell(n)$ bound the lengths of the paths that $\calA$ outputs, and let $W_x$ be the words of length at most $\ell(n)$, with $\varphi_x(\omega) := \omega\cdot x$.  Every $y\in\calC_{k(x)}$ is reached by the word $\calA(x,y)$, so $\varphi_x(W_x) = \calC_{k(x)}$ and $(W_x,\varphi_x)$ is an index set.  The map $u_x(y) := \calA(x,y)$ is a section with $\mathrm{im}\,u_x = \Ima_x$.  Both $\varphi_x$ and $u_x$ are computable in polynomial time, so \Cref{prop:index-set} applies.
\end{proof}

A path-finder thus turns orbit state preparation into preparing $\ket{\Ima_x}$, the uniform superposition over the paths it outputs from the seed, a problem that is equally hard.

\paragraph{Why this gives no forgery.}
Write $\calC:=\calC_{k(x)}$.  Running $\calA$ coherently and then uncomputing $y$ from the path, by applying it to $x$, acts as
\[
    \sum_{y\in\calC}\alpha_y\ket y\ket0 \;\longmapsto\; \sum_{y\in\calC}\alpha_y\ket y\ket{\calA(x,y)} \;\longmapsto\; \sum_{y\in\calC}\alpha_y\ket0\ket{\calA(x,y)}\,,
\]
which is $\ket0\ket{\Ima_x}$ exactly when $\alpha_y = |\calC|^{-1/2}$ for all $y$, that is, when the input is the orbit state itself.  The natural superposition over $\calC$ that can be prepared without it is the endpoint of a random walk from $x$, and that comes entangled with the walk's own record.  Recorded, the walk leaves the element register with accepted weight $1/|\calC|$, no more than the seed alone (\Cref{prop:walk-index}(iv)).  Disentangling it in the manner of \Cref{lem:uncompute} needs relabellings of the walk's fibres, which exist approximately but are not known to be efficient (\Cref{prop:walk-index}(ii) and~(iii)).

\paragraph{Randomized path-finders.}
The hypothesis that $\calA$ is deterministic costs little for a classical path-finder that succeeds on every pair with probability at least $2/3$.  Its outputs can be checked, so repetition reduces its failure probability to $2^{-3n}$ on each pair.  There are at most $2^{2n}$ pairs of objects, so a random string fixed in advance makes it fail on some pair with probability at most $2^{2n}\cdot2^{-3n} = 2^{-n}$.

\section{Amplification and bounds for orbit state preparation}\label{sec:osp-def}

This section studies orbit state preparation itself.  Verification amplifies any attack whose output stays in the orbit of its seed (\Cref{sec:amplification}), two generic attacks bound the problem from above (\Cref{sec:generic}), and for rapidly mixing schemes an $\NP$ oracle solves it (\Cref{sec:oracle-taxonomy}).

\subsection{Amplification}\label{sec:amplification}

An attack whose output has inverse-polynomial accepted weight and stays in the orbit of its seed can be boosted to the orbit state, because the verifier implements a rotation about the orbit states.  For an orbit $\calC_k$ with serial number $p$, call the other orbits with serial number $p$ its \emph{co-orbits}.

\begin{lemma}[Amplification by verification]\label{lem:amplification}
Let $x\in\calC_k$ and $p:=I(x)$.  Let $D$ be a polynomial-size circuit, and let $\sigma$ be the reduced state of $D\ket x\ket{0^m}$ on the element register.  Suppose that
\begin{enumerate}[label=\textup{(\alph*)},leftmargin=*,nosep]
\item $\sigma$ is supported on $\calC_k$, or $I$ is injective on orbits; and
\item $\aw_p(\sigma)\ge\eta$.
\end{enumerate}
Let $0<\Delta_\ast\le\Delta$ be a known lower bound on the spectral gap, and let $\delta>0$.  Then there is a circuit $\mathrm{Amp}$, depending on $x$ only through its input register, with
\[
    \bigl\|\mathrm{Amp}\ket x\ket0-\ket{\psi_k}\otimes\ket{a_x}\bigr\| \;\le\; \delta \qquad\text{for some state }\ket{a_x}\,.
\]
It makes $O(\eta^{-1/2}\log(1/\delta))$ calls to $D$, to $D^\dagger$ and to the coherent verifier, which runs $\Ver$ with its measurements deferred for $O(\Delta_\ast^{-1}\log(1/\delta\eta))$ rounds (\Cref{app:amplification}).  In particular, if the scheme mixes rapidly and~(a) and~(b) hold with $\eta\ge1/\poly(n)$ for every $x$ in a family of seeds $G$, then orbit state preparation is solvable on $G$.
\end{lemma}
\begin{proof}
\Cref{app:amplification} gives the proof.  It is fixed-point amplitude amplification that alternates a rotation about the output of $D$ with a rotation about the states whose element register holds the orbit state, and a coherent verifier implements the second rotation.  For the last sentence, take $\Delta_\ast=n/r$ and $\delta=2^{-n}$.
\end{proof}

\subsection{Generic attacks}\label{sec:generic}

Two attacks work for every scheme and bound orbit state preparation from above.  The first amplifies the seed itself, whose accepted weight $1/|\calC_{k(x)}|$ is unknown, by trying the guesses $\eta = 2^{-j}$ in turn.  The second rebuilds a banknote from its serial number alone, by amplitude amplification on the mint.

\begin{proposition}[Generic attacks]\label{prop:generic-attacks}
Let $(X,S,I)$ be an invariant money scheme, and let $\mathrm{wt}(p) = |P_p|/|X|$ be the minting weight of $p$ (\Cref{def:sparse}).
\begin{enumerate}[label=\textup{(\roman*)},leftmargin=*,nosep]
\item If the scheme mixes rapidly, there is an algorithm that, on every seed $x$, halts within time $\sqrt{|\calC_{k(x)}|}\cdot r\cdot\poly(n)$ except with probability $2^{-\Omega(n)}$, and outputs a state whose element register is within trace distance $2^{-\Omega(n)}$ of $\ket{\psi_{k(x)}}$.
\item There is an algorithm that, on every serial number $p$ in the image of $I$ and with no seed, outputs $\ket{\$_p}$ in expected time $\mathrm{wt}(p)^{-1/2}\cdot\poly(n)$.
\end{enumerate}
\end{proposition}
\begin{proof}
(i) Write $\calC := \calC_{k(x)}$ and $p := I(x)$.  With $D$ the identity, $\sigma = \ket x\!\bra x$ satisfies hypothesis~(a) of \Cref{lem:amplification}, and $\aw_p(\sigma) = |\braket{\psi_{k(x)}}{x}|^2 = 1/|\calC|$.  For $j = 1,\dots,n$ let $\mathrm{Amp}_j$ be the circuit of \Cref{lem:amplification} with $\eta = 2^{-j}$, $\delta = 2^{-2n}$ and $\Delta_\ast = n/r$.  The algorithm is
\begin{quote}
for $j = 1,\dots,n$: prepare $\mathrm{Amp}_j\ket x\ket0$, apply $\Ver(p,\cdot)$, and halt with the register if $\Ver$ accepts.
\end{quote}
Let $j^\ast$ be the least $j$ with $2^{-j}\le1/|\calC|$.  Then $j^\ast\le n$, because $|\calC|\le2^n$, and $2^{j^\ast}<2|\calC|$.  At $j^\ast$ hypothesis~(b) holds, so the output is within $2^{-2n}$ of $\ket{\psi_{k(x)}}\otimes\ket{a_x}$, and it is accepted except with probability $2^{-\Omega(n)}$ (\Cref{lem:spectral-overlap}).  Run $j$ makes $O(2^{j/2}n)$ calls to the verifier, of $O(r)$ rounds each, so halting by run $j^\ast$ takes time
\[
    \sum_{j=1}^{j^\ast}O(2^{j/2}n)\cdot O(r)\cdot\poly(n) \;=\; O\bigl(2^{j^\ast/2}\bigr)\cdot r\cdot\poly(n) \;=\; \sqrt{|\calC|}\cdot r\cdot\poly(n)\,.
\]
For each run, \Cref{lem:proj-implementation} gives
\[
    \Pr\bigl[\Ver\text{ accepts and its register lies outside the range of }\Pi_1^{(p)}\bigr] \;\le\; (1-\Delta)^{2r} \;\le\; e^{-2n}\,,
\]
and on $\calC$ the range of $\Pi_1^{(p)}$ is spanned by $\ket{\psi_{k(x)}}$.  Summing over the at most $n$ runs,
\[
    \Pr\bigl[\text{the algorithm halts with its register outside }\mathrm{span}\{\ket{\psi_{k(x)}}\}\bigr] \;\le\; n\,e^{-2n}\,,
\]
so the output register is within trace distance $2^{-\Omega(n)}$ of $\ket{\psi_{k(x)}}$.

(ii) Without its final measurement, the mint prepares
\[
    |X|^{-1/2}\sum_{y\in X}\ket y\ket{I(y)} \;=\; \sum_{p'}\mathrm{wt}(p')^{1/2}\,\ket{\$_{p'}}\ket{p'}\,.
\]
The event that the second register reads $p$ has probability $\mathrm{wt}(p)$, and the first register then holds $\ket{\$_p}$.  Amplitude amplification with an unknown success probability~\cite{BHMT02} reaches the event in expected $O(\mathrm{wt}(p)^{-1/2})$ iterations, each using the mint and its inverse once, and on success the first register holds $\ket{\$_p}$ exactly.
\end{proof}

\paragraph{The banknote attack and dispersion.}
Attack~(ii) outputs the banknote $\ket{\$_p}$, which is an orbit state only when $I$ is injective on orbits.  Under dispersion (\Cref{def:sparse}), its expected time $\mathrm{wt}(p)^{-1/2}$ is superpolynomial for every $p$.

\subsection{An upper bound from approximate counting}\label{sec:oracle-taxonomy}

For every rapidly mixing scheme, an $\NP$ oracle solves orbit state preparation.  The uniform superposition over a set can be built bit by bit from its prefix marginals~\cite{GR02,AharonovTaShma03}, which is Aaronson's state synthesis algorithm~\cite{Aaronson16,INNRY22}.  An $\NP$ oracle supplies these marginals only up to a multiplicative factor, and \Cref{lem:count-complete} shows that the construction still works with approximate marginals.  Membership in an orbit is in $\NP$ because rapid mixing bounds the diameter of every orbit, as Lutomirski observes for component mixers~\cite[\S4.1]{Lutomirski11}.

\begin{definition}[Approximate prefix-counting oracle]\label{def:count-oracle}
For non-empty $Q \subseteq \{0,1\}^n$ and a prefix $b \in \{0,1\}^{\leq n}$ write $\cnt_Q(b) := |\{y \in Q : y_{1\ldots|b|} = b\}|$.  An \emph{$\epsilon$-approximate prefix-counting oracle} for $Q$ returns on query $b$ a value $\widetilde\cnt(b)$ that depends on $b$ alone and satisfies
\begin{equation}\label{eq:count-oracle}
    (1-\epsilon)\,\cnt_Q(b) \;\leq\; \widetilde\cnt(b) \;\leq\; (1+\epsilon)\,\cnt_Q(b)\,.
\end{equation}
It is accessed as the unitary $\ket b\ket z \mapsto \ket b\ket{z \oplus \widetilde\cnt(b)}$.
\end{definition}

Unitary access lets the oracle be queried on a superposition of prefixes.

\begin{lemma}[State preparation from approximate prefix counts]\label{lem:count-complete}
Let $\zeta=(1+\epsilon)/(1-\epsilon)$. There is a circuit of size $\poly(n)$ making $4n$ queries to an $\epsilon$-approximate prefix-counting oracle for $Q$ that outputs $\ket{\tilde\psi_Q}\otimes\ket0$, where $\ket{\tilde\psi_Q}$ is supported on $Q$ and $|\braket{Q}{\tilde\psi_Q}|\ge\zeta^{-n/2}$, which is at least $1/\sqrt2$ at $\epsilon=1/(4n)$.
\end{lemma}
\begin{proof}
Build the string bit by bit.  For a prefix $b$ of length $j<n$ and $a\in\{0,1\}$, write
\[
    \tilde\chi(b,a) \;:=\; \frac{\widetilde\cnt(ba)}{\widetilde\cnt(b0)+\widetilde\cnt(b1)}\,,\qquad \chi(b,a) \;:=\; \frac{\cnt_Q(ba)}{\cnt_Q(b)}\,.
\]
Step $j+1$ acts on the prefix register and a fresh qubit as
\[
    \ket b\ket0 \;\longmapsto\; \ket b\bigl(\tilde\chi(b,0)^{1/2}\ket0 + \tilde\chi(b,1)^{1/2}\ket1\bigr)\,,
\]
using two queries to compute $\widetilde\cnt(b0)$ and $\widetilde\cnt(b1)$ into ancillas, a controlled rotation, and two more queries to uncompute them, which is possible because $\widetilde\cnt$ is a function of $b$.  After $n$ steps the output is
\[
    \ket{\tilde\psi_Q} \;=\; \sum_{y}\tilde\beta(y)^{1/2}\ket y\,,\qquad \tilde\beta(y) \;:=\; \prod_{j=1}^{n}\tilde\chi(y_{<j},y_j)\,.
\]
A branch with $\cnt_Q(ba)=0$ gets amplitude $0$, so the support stays in $Q$, and a prefix in the support has $\cnt_Q(b)>0$, so the denominator of $\tilde\chi$ is positive.  Since $t_a/(t_a+t_{\bar a})$ is increasing in $t_a$ and decreasing in $t_{\bar a}$, the $(1\pm\epsilon)$ errors of the counts give $\zeta^{-1}\chi(b,a)\le\tilde\chi(b,a)\le\zeta\,\chi(b,a)$.  For $y\in Q$ the true factors telescope to $\prod_j\chi(y_{<j},y_j) = 1/|Q|$, so $\tilde\beta(y)\ge\zeta^{-n}/|Q|$.  All amplitudes are non-negative, so
\[
    |\braket{Q}{\tilde\psi_Q}| \;=\; \sum_{y\in Q}\Bigl(\frac{\tilde\beta(y)}{|Q|}\Bigr)^{1/2} \;\ge\; \sum_{y\in Q}\frac{\zeta^{-n/2}}{|Q|} \;=\; \zeta^{-n/2}\,.
\]
At $\epsilon=1/(4n)$, $\zeta^{-n}\ge(1-2\epsilon)^n\ge1-2\epsilon n=\tfrac12$.
\end{proof}

\begin{proposition}[Orbit state preparation with an $\NP$ oracle]\label{prop:np-forges}
Let $(X,S,I)$ be a rapidly mixing invariant money scheme. Then orbit state preparation is solvable by a quantum polynomial-time algorithm with oracle access to an $\NP$ language. Consequently, if $\NP\subseteq\BQP$, orbit state preparation is solvable.
\end{proposition}
\begin{proof}
Write $\calC := \calC_{k(x)}$ and $p := I(x)$.  With $Q = \calC$ and $\epsilon = 1/(4n)$, the output of \Cref{lem:count-complete} is supported on $\calC$ and has accepted weight $|\braket{\calC}{\tilde\psi_\calC}|^2\ge1/2$, so \Cref{lem:amplification} applies with $\eta = 1/2$.  It therefore remains to implement an $\epsilon$-approximate prefix-counting oracle for $\calC$.

\emph{Diameter.}  Let $T := r+1$.  For $y\in\calC$, by the spectral decomposition of $\hat B$ on $\calC$ and $\Delta T > \Delta r\ge n$,
\[
    \Bigl|(\hat B^T)_{yx}-\frac1{|\calC|}\Bigr| \;\le\; (1-\Delta)^T \;<\; e^{-n} \;<\; 2^{-n} \;\le\; \frac1{|\calC|}\,,
\]
so $(\hat B^T)_{yx}>0$.  This entry is the probability that $T$ lazy steps from $x$ end at $y$, so every $y\in\calC$ is $\omega\cdot x$ for some word $\omega\in S^{\le T}$.

\emph{Membership in $\NP$.}  Hence
\[
    \cnt_\calC(b) \;=\; \bigl|\{\,y\in\{0,1\}^n : y_{1\ldots|b|} = b \text{ and } \omega\cdot x = y \text{ for some } \omega\in S^{\le T}\,\}\bigr|\,,
\]
the size of a set with a polynomial-time verifiable witness $\omega$.

\emph{Estimation.}  Sipser--Stockmeyer hashing~\cite{Sipser83,Stockmeyer83} estimates the size of such a set within a constant factor $c$, using $\poly(n)$ random bits and $\poly(n)$ queries to the $\NP$ language $\{(b,h) : \exists y\,\exists\omega\,[\omega\cdot x = y \wedge y_{1\ldots|b|} = b \wedge h(y)=0]\}$.  Applied to the $N$-fold product of the set, whose size is $\cnt_\calC(b)^N$, and followed by an $N$-th root, it gives an estimate $\widetilde\cnt(b)$ with
\[
    c^{-1/N}\cnt_\calC(b) \;\le\; \widetilde\cnt(b) \;\le\; c^{1/N}\cnt_\calC(b)\,,
\]
which satisfies equation~\eqref{eq:count-oracle} at $N = \Theta(1/\epsilon)$.

\emph{Determinism.}  Take the median of $\Theta(n)$ independent estimates, so that it fails on a fixed prefix with probability below $2^{-n-1}/3$.  There are fewer than $2^{n+1}$ prefixes, so one random string drawn at the outset is good for every prefix except with probability at most
\[
    2^{n+1}\cdot\frac{2^{-n-1}}{3} \;=\; \frac13\,.
\]
With a good string fixed, $\widetilde\cnt$ is a function of $b$, so it can be run coherently and uncomputed, as \Cref{lem:count-complete} requires.

\emph{Assembly.}  The algorithm repeats, at most $O(n)$ times: draw a string, apply \Cref{lem:count-complete} and then $\mathrm{Amp}$ of \Cref{lem:amplification} with $\delta = 2^{-2n}$, apply $\Ver(p,\cdot)$, and halt on acceptance.  A prefix with no extension in $\calC$ has no witness and is estimated as $0$.  If both estimates at a node are $0$, which a good string never produces, the circuit writes $x$ on that branch, so every state stays supported on $\calC$.  With a good string an attempt is accepted except with probability $2^{-\Omega(n)}$, so all $O(n)$ attempts fail with probability $2^{-\Omega(n)}$.  As in the proof of \Cref{prop:generic-attacks}(i), \Cref{lem:proj-implementation} puts the output register within trace distance $t = 2^{-\Omega(n)}$ of $\ket{\psi_{k(x)}}$.  Deferring the measurements gives a circuit of polynomial size with output $\ket\Phi$, whose element register has fidelity $F := \bra{\psi_{k(x)}}\rho\ket{\psi_{k(x)}}\ge1-t$.  With $\ket{a_x}$ the normalized vector $(\bra{\psi_{k(x)}}\otimes\openone)\ket\Phi$, which depends on $x$ only,
\[
    \bigl\|\ket\Phi-\ket{\psi_{k(x)}}\otimes\ket{a_x}\bigr\|^2 \;=\; 2-2\sqrt F \;\le\; 2t\,,
\]
so the circuit solves orbit state preparation.  If $\NP\subseteq\BQP$, replace the oracle by a quantum polynomial-time algorithm run coherently with error reduction.
\end{proof}

\section{The preparation assumption and security}\label{sec:reduction}

For the knot scheme no polynomial-time algorithm for orbit state preparation is known.  When condition~\condH{} does not hold, no reduction of it to path-finding is known either (\Cref{sec:where-not}).  This section presents an average-case assumption for orbit state preparation we call the \emph{preparation assumption} (\Cref{sec:the-assumption}).  The assumption is equivalent to security against adversaries that measure their banknote first (\Cref{sec:collapse}), and together with a second assumption, which concerns adversaries that keep the banknote coherent, it is equivalent to security (\Cref{sec:transfer}).  \Cref{sec:transfer-necessity} shows that a natural class of reductions cannot derive the second assumption from the first.

\subsection{The preparation assumption}\label{sec:the-assumption}

\begin{assumption}[Preparation assumption]\label{ass:synth}
For every quantum polynomial-time algorithm $\calA$,
\[
    \adv_\calA(n)\;:=\;\mathbb E_{(p,x)\leftarrow\mathcal M_n}\bigl[\aw_p(\sigma_\calA(p,x))\bigr]\;=\;\negl(n)\,,
\]
where $\sigma_\calA(p,x)$ is the reduced state of the register that $\calA$ outputs.
\end{assumption}

\paragraph{Choices in the assumption.}
\begin{itemize}[leftmargin=*,nosep]
\item $\sigma_\calA(p,x)$ is a reduced state, because an adversary guarantees only that one register passes verification, not that its ancillas are clean.
\item The assumption bounds accepted weight, not the overlap $\bra{\$_p}\sigma\ket{\$_p}$ with the banknote, because a forger need not reproduce the banknote.
\item Accepted weight counts every orbit with serial number $p$,
\[
    \aw_p(\sigma) \;=\; \sum_{k :\, I(\calC_k)=p}\bra{\psi_k}\sigma\ket{\psi_k}\,,
\]
including the co-orbits of the seed's orbit, because verification accepts all of these orbit states alike (\Cref{lem:eigenspace}).  It is the quantity that a challenger can estimate (\Cref{prop:falsifiable}).
\end{itemize}

The assumption is an average-case, inverse-polynomial form of orbit state preparation, in the following sense.

\begin{proposition}[Preparation as average-case orbit state preparation]\label{prop:prep-vs-osp}
\leavevmode
\begin{enumerate}[label=\textup{(\roman*)},leftmargin=*,nosep]
\item If a family of seeds of non-negligible minting weight has a solver for orbit state preparation, then the preparation assumption fails.
\item Suppose the scheme mixes rapidly, and some quantum polynomial-time $\calA$ has $\adv_\calA(n)\ge1/q(n)$ for infinitely many $n$, where either $\calA$'s output lies in the orbit of its seed or $I$ is injective on orbits.  Then for those $n$, orbit state preparation is solvable on a family of seeds of minting weight at least $1/(2q(n))$.
\end{enumerate}
\end{proposition}
\begin{proof}
(i) Let $\Prep$ solve orbit state preparation on $G$ with error $\mu$, and let $\calA$ run $\Prep$ on its seed.  For $x\in G_n$, the reduced state $\sigma$ of the output register is within trace distance $\mu$ of $\ket{\psi_{k(x)}}$, because the partial trace is a channel.  So $\aw_p(\sigma)\ge\bra{\psi_{k(x)}}\sigma\ket{\psi_{k(x)}}\ge1-\mu$, and
\[
    \adv_\calA(n) \;\ge\; \Pr_{(p,x)\leftarrow\mathcal M_n}[x\in G_n]\,\bigl(1-\mu(n)\bigr)\,,
\]
which is non-negligible.

(ii) Write $a(x) := \aw_{I(x)}(\sigma_\calA(I(x),x))$ and $G_n := \{x : a(x)\ge1/(2q(n))\}$.  Since $a(x)\le1$ and $\mathbb E_x[a(x)] = \adv_\calA(n)\ge1/q(n)$,
\[
    \frac1{q(n)} \;\le\; \Pr_x[x\in G_n]\cdot1 + \frac1{2q(n)}\,, \qquad\text{so}\qquad \frac{|G_n|}{|X_n|}\;\ge\;\frac1{2q(n)}\,.
\]
For $x\in G_n$, the circuit of $\calA$ satisfies hypotheses~(a) and~(b) of \Cref{lem:amplification} with $\eta = 1/(2q(n))$, so by the last sentence of that lemma orbit state preparation is solvable on $G$.
\end{proof}

When $I$ is not injective on orbits, part~(ii) needs $\calA$'s output to stay in the orbit of its seed: a violation may instead put its weight on co-orbits, and then no solver follows.

The assumption is falsifiable, because the public verifier estimates accepted weight.

\begin{proposition}[Falsifiability of the preparation assumption]\label{prop:falsifiable}
Suppose the scheme mixes rapidly.  For an algorithm $\calA$, let a challenger sample $(p,x)\leftarrow\mathcal M_n$, run $\calA(p,x)$, and run $\Ver(p,\cdot)$ on the register that $\calA$ returns.  Then the challenger accepts with probability between $\adv_\calA(n)$ and $\adv_\calA(n)+e^{-2n}$.  So the preparation assumption holds if and only if no quantum polynomial-time $\calA$ makes this efficient challenger accept with non-negligible probability, and it is falsifiable in the sense of Naor~\cite{Naor03}.
\end{proposition}
\begin{proof}
The challenger is efficient, because $\mathcal M_n$ is efficiently samplable (\Cref{def:measured-mint}) and $\Ver$ is public.  It accepts with probability $\mathbb E[\mathrm{Acc}_p(\sigma_\calA(p,x))]$.  By \Cref{lem:spectral-overlap} at the verifier's $r$, $\aw_p(\sigma)\le\mathrm{Acc}_p(\sigma)\le\aw_p(\sigma)+(1-\Delta)^{2r}\le\aw_p(\sigma)+e^{-2n}$ for every state $\sigma$, and taking expectations over $(p,x)$ gives
\[
    \adv_\calA(n) \;\le\; \Pr[\text{the challenger accepts}] \;\le\; \adv_\calA(n)+e^{-2n}\,.
\]
\end{proof}

\subsection{Measuring first}\label{sec:collapse}

An adversary that measures its banknote is left with the serial number $p$ and one element $x$ of the level set.  The basis state $\ket x$ has accepted weight $|\braket{\psi_{k(x)}}{x}|^2 = 1/|\calC_{k(x)}|$, so the adversary must build both of its outputs from $(p,x)$, and the preparation assumption says that it cannot build even one.

\begin{definition}[Measure-first adversaries]\label{def:measure-first}
$\Meas$ measures the challenge register in the computational basis. For an adversary $\calA$, $\calA\circ\Meas$ applies $\Meas$ and then runs $\calA$. An adversary is \emph{measure-first} if it has the form $\calA\circ\Meas$, and $\win^{\Meas}_\calA:=\win_{\calA\circ\Meas}$.
\end{definition}

\begin{theorem}[Measure-first security and the preparation assumption]\label{thm:collapse-equiv}
Let $(X,S,I)$ be a rapidly mixing invariant money scheme.  Then some quantum polynomial-time $\calA$ has non-negligible $\win^{\Meas}_\calA$ if and only if the preparation assumption fails.  The implication from a failure of the preparation assumption to such an $\calA$ holds without rapid mixing.
\end{theorem}
\begin{proof}
Measuring $\ket{\$_p}$ in the computational basis returns an element uniform on $P_p$, so the pair $(p,x)$ that a measure-first adversary $\calA\circ\Meas$ works on is distributed as $\mathcal M_n$ (\Cref{def:measured-mint}).  Hence
\begin{equation}\label{eq:collapse-classical}
    \win^{\Meas}_{\calA}(n) \;=\; \Pr_{(p,x)\leftarrow\mathcal M_n}\bigl[\,\text{both output registers of } \calA(p,\ket{x}) \text{ accept at } p\,\bigr]\,,
\end{equation}

\emph{A winning measure-first adversary violates the assumption.} If $\win^{\Meas}_\calA\ge1/q(n)$ for infinitely many $n$, then by equation~\eqref{eq:collapse-classical} the second output of $\calA(p,\ket x)$, in reduced state $\sigma$, passes with expected probability at least $1/q$.  At the verifier's $r$, \Cref{lem:spectral-overlap} gives $\mathrm{Acc}_p(\sigma)\le\aw_p(\sigma)+(1-\Delta)^{2r}\le\aw_p(\sigma)+e^{-2n}$, so
\[
    \mathbb E\bigl[\aw_p(\sigma)\bigr] \;\ge\; \mathbb E\bigl[\mathrm{Acc}_p(\sigma)\bigr]-e^{-2n} \;\ge\; 1/q-e^{-2n}\,.
\] The algorithm that runs $\calA(p,\ket x)$ and keeps $\mathsf R_2$ is quantum polynomial-time on $(p,x)$.

\emph{A violation yields a winning measure-first adversary (rapid mixing not used).} If $\adv_\calB\ge1/\poly$ infinitely often, let $\calA(p,\ket x)$ run $\calB(p,x)$ twice independently. For fixed $(p,x)$ the two outputs are independent, and each passes with probability $\mathrm{Acc}_p(\sigma_\calB)\ge\aw_p(\sigma_\calB)$ by the lower bound of \Cref{lem:spectral-overlap}, which holds for every $r$.  So by Jensen's inequality both pass with probability
\[
    \mathbb E\bigl[\mathrm{Acc}_p(\sigma_\calB)^2\bigr] \;\ge\; \mathbb E\bigl[\aw_p(\sigma_\calB)^2\bigr] \;\ge\; \bigl(\mathbb E\,\aw_p(\sigma_\calB)\bigr)^2 \;=\; \adv_\calB^2\,.
\]
\end{proof}

\subsection{The transfer assumption and the decomposition of security}\label{sec:transfer}

\Cref{thm:collapse-equiv} settles the game for adversaries that measure first.  For the others we assume that every adversary does at most a polynomial factor better than some adversary that measures first.

\begin{assumption}[Transfer assumption]\label{ass:transfer}
There is a polynomial $q$ such that for every quantum polynomial-time adversary $\calA$ there are a quantum polynomial-time measure-first adversary $\calB$ and a negligible $\nu_\calA$ with
\[
    \win_\calA(n)\;\le\; q(n)\cdot\win_\calB(n)+\nu_\calA(n)\,.
\]
\end{assumption}

The measure-first adversary $\calB$ may differ from $\calA\circ\Meas$.  An adversary that forwards its banknote as one of its outputs loses that output if it measures first, so comparing $\calA$ only with $\calA\circ\Meas$ would make the assumption too strong.

\begin{theorem}[Security decomposition]\label{thm:general}
Let $(X,S,I)$ be an invariant money scheme.
\begin{enumerate}[label=\textup{(\roman*)},leftmargin=*,nosep]
\item If the preparation assumption fails, the scheme is not secure.
\item If the scheme mixes rapidly, it is not secure if and only if the preparation assumption or the transfer assumption fails.
\end{enumerate}
\end{theorem}
\begin{proof}
(i) If the preparation assumption fails, the half of \Cref{thm:collapse-equiv} that does not use rapid mixing gives an efficient $\calA$ with $\win_{\calA\circ\Meas}$ non-negligible.  The adversary $\calA\circ\Meas$ is efficient, so the scheme is not secure.

(ii) If the preparation assumption fails, the scheme is not secure by part~(i).  If the transfer assumption fails, its negation at $q\equiv1$ gives an efficient $\calA$ such that, for every efficient measure-first $\calB$ and every negligible $\nu$, $\win_\calA(n)>\win_\calB(n)+\nu(n)$ for some $n$.  If $\win_\calA$ were negligible, $\nu:=\win_\calA$ would give $\win_\calA(n)\le\win_\calB(n)+\nu(n)$ for every $n$.  So $\win_\calA$ is not negligible, and the scheme is not secure.

Conversely, suppose the scheme is not secure, so some efficient $\calA$ has non-negligible $\win_\calA$, and suppose the transfer assumption holds.  It gives an efficient measure-first $\calB=\calB'\circ\Meas$ and a negligible $\nu_\calA$ with $\win_\calA\le q\cdot\win_\calB+\nu_\calA$.  Since $q$ is a polynomial, $\win^{\Meas}_{\calB'}=\win_\calB$ is not negligible, and by \Cref{thm:collapse-equiv} the preparation assumption fails.
\end{proof}

\subsection{A barrier for serial-fixing reductions}\label{sec:transfer-necessity}

Unlike the preparation assumption, the transfer assumption is not known to be falsifiable.  One would therefore like to derive it from the preparation assumption, and by \Cref{thm:general}, for rapidly mixing schemes that satisfy the preparation assumption, this is the same as proving security.  We show that reductions that call the adversary only at the serial number they are given cannot do so.

\begin{definition}[Serial-fixing reduction]\label{def:serial-fixing}
Let $\calR$ be a quantum polynomial-time algorithm that receives $(p,x)\leftarrow\mathcal M_n$ and outputs one register, in reduced state $\sigma_{\calR^\calA}(p,x)$. It accesses an adversary $\calA$ for the game of \Cref{def:security} only through a unitary $U_\calA=\sum_{p'}\ket{p'}\!\bra{p'}\otimes U_{\calA,p'}$ implementing $\calA$ on a serial-number register, a challenge register and further registers.  Precisely, $\calR$ is a polynomial-size circuit containing at most $t=t(n)=\poly(n)$ instances of $U_\calA$ and $U_\calA^\dagger$, and all of their input registers are wires of $\calR$. Say that $\calR$ is
\begin{enumerate}[label=\textup{(\roman*)},leftmargin=*,nosep]
\item \emph{serial-fixing} if the serial-number register is in the basis state $\ket p$ at every instance of $U_\calA$ and $U_\calA^\dagger$;
\item \emph{sound} if for every $\calA$, efficient or not, $\win_\calA(n)\ge1/\poly(n)$ infinitely often implies that $\mathbb E_{(p,x)}[\aw_p(\sigma_{\calR^\calA}(p,x))]\ge1/\poly(n)$ infinitely often.
\item a \emph{fully black-box derivation of the transfer assumption from the preparation assumption} if it is sound.  This is the right notion because, for rapidly mixing schemes that satisfy the preparation assumption, the transfer assumption holds exactly when the scheme is secure (\Cref{thm:general}).
\end{enumerate}
\end{definition}

\paragraph{What the definition allows.}
\begin{itemize}[leftmargin=*,nosep]
\item Access through $U_\calA$ is at least as strong as access to $\calA$ as a channel.
\item $\calR$ may prepare challenges in superposition, rewind, let calls interfere, and initialize every register it passes to $U_\calA$, including the adversary's workspace.  Any black-box procedure that tries to extract a path from the adversary is of this kind.
\item Soundness quantifies over inefficient adversaries, as is standard for fully black-box reductions.
\end{itemize}

\begin{theorem}[Barrier for serial-fixing reductions]\label{thm:transfer-barrier}
Let $(X,S,I)$ be an invariant money scheme satisfying the preparation assumption.  Then no serial-fixing reduction is sound, so none is a fully black-box derivation of the transfer assumption from the preparation assumption.  Specifically, for each serial number $p'$ let $\Syn_{p'}$ be a unitary on one register with $\Syn_{p'}\ket0=\ket{\$_{p'}}$, and let $\calF^\ast$ be the adversary implemented by
\[
U^\ast:=\sum_{p'}\ket{p'}\!\bra{p'}\otimes\Bigl(\Pi_1^{(p')}\otimes \Syn_{p'}+\bigl(\openone-\Pi_1^{(p')}\bigr)\otimes\openone\Bigr)
\]
on its serial-number register, its challenge register and a second output register. Then
\begin{enumerate}[label=\textup{(\roman*)},leftmargin=*,nosep]
\item $\calF^\ast$ wins the game of \Cref{def:security} with probability $1$; and
\item for every serial-fixing reduction $\calR$, $\mathbb E_{(p,x)}[\aw_p(\sigma_{\calR^{\calF^\ast}}(p,x))]=\negl(n)$.
\end{enumerate}
\end{theorem}
\begin{proof}
(i) Since $\Pi_1^{(p)}\ket{\$_p}=\ket{\$_p}$,
\[
    U^\ast\bigl(\ket p\ket{\$_p}\ket0\bigr) \;=\; \ket p\otimes\Pi_1^{(p)}\ket{\$_p}\otimes\Syn_p\ket0 \;+\; \ket p\otimes\bigl(\openone-\Pi_1^{(p)}\bigr)\ket{\$_p}\otimes\ket0 \;=\; \ket p\ket{\$_p}\ket{\$_p}\,,
\]
and both registers pass.

(ii) The idea is to compare $\calR$ run with $\calF^\ast$ against $\calR$ run with the adversary that does nothing.  The two differ only on the part of each call's input that has accepted weight, and the preparation assumption makes that part negligible.  Let $\tilde\calF$ be the adversary that does nothing, $U_\bot:=\openone$.  Then $\calR^{\tilde\calF}$ is quantum polynomial-time on $(p,x)$. Assume without loss of generality that $\calR$ defers its measurements. At every call the serial-number register holds $p$, where $U^\ast-U_\bot=\Pi_1^{(p)}\otimes(\Syn_p-\openone)$ and $U^{\ast\dagger}-U_\bot^\dagger=\Pi_1^{(p)}\otimes(\Syn_p^\dagger-\openone)$, so both differences map any vector $\ket\vartheta$ to a vector of norm at most $2\|(\Pi_1^{(p)}\otimes\openone)\ket\vartheta\|$. Let $\ket{\Theta_{i-1}}$ be the state of the run $\calR^{\tilde\calF}$ just before its $i$-th call and $\sigma_i$ the reduced state of its challenge register. Replacing the calls one at a time \cite{BBBV97},
\[
\bigl\|\mathrm{out}(\calR^{\calF^\ast})-\mathrm{out}(\calR^{\tilde\calF})\bigr\|_{\mathrm{tr}}\le4\sum_{i=1}^t\bigl\|(\Pi_1^{(p)}\otimes\openone)\ket{\Theta_{i-1}}\bigr\|=4\sum_{i=1}^t\aw_p(\sigma_i)^{1/2}.
\]
Let $\calB$ be the algorithm that on $(p,x)$ draws $i$ uniformly from $\{1,\dots,t\}$, runs $\calR^{\tilde\calF}$ up to its $i$-th call and outputs the challenge register. It is quantum polynomial-time, so $\adv_\calB=t^{-1}\sum_i\mathbb E\,\aw_p(\sigma_i)=\negl(n)$ by \Cref{ass:synth}.  Accepted weight changes by at most half the trace norm of a difference, and $\calR^{\tilde\calF}$ is quantum polynomial-time on $(p,x)$, so taking expectations over $(p,x)$,
\begin{align*}
    \mathbb E\,\aw_p\bigl(\sigma_{\calR^{\calF^\ast}}\bigr)
    &\;\le\; \mathbb E\,\aw_p\bigl(\sigma_{\calR^{\tilde\calF}}\bigr) + 2\sum_{i=1}^t\mathbb E\,\aw_p(\sigma_i)^{1/2}\\
    &\;\le\; \adv_{\calR^{\tilde\calF}} + 2\Bigl(t\sum_{i=1}^t\mathbb E\,\aw_p(\sigma_i)\Bigr)^{1/2}
    \;=\; \adv_{\calR^{\tilde\calF}} + 2t\,\adv_\calB^{1/2} \;=\; \negl(n)\,,
\end{align*}
using Jensen's inequality and Cauchy--Schwarz in the second line, and \Cref{ass:synth} for $\calR^{\tilde\calF}$ and $\calB$.  This is part~(ii).  By part~(i), $\calF^\ast$ witnesses that $\calR$ is not sound.
\end{proof}

The proof follows the separation technique of Gentry and Wichs~\cite{GW11}: an inefficient adversary that breaks the scheme, and an efficient one that no reduction of the class can tell apart from it.

\paragraph{Reductions outside the class.}
\begin{itemize}[leftmargin=*,nosep]
\item A reduction that calls the adversary at another serial number.  The reduction from quantum lightning (\Cref{def:lightning}) to unforgeability is of this kind: it runs $\Mint(1^n)$ itself, hands $(p,\ket{\$_p})$ to the forger, and outputs $p$ with the forger's two registers~\cite[Thm.~12]{LMZ22}.
\item A reduction with non-black-box access to the adversary, such as one that extracts information from the adversary's final state.
\item A reduction that is sound only against efficient adversaries, since $\calF^\ast$ is inefficient.
\item A reduction from a problem whose instance is itself a banknote, such as the duplication problem of Montgomery and Sharif~\cite[Thm.~9.2]{MS24}, which does not start from $(p,x)$.
\end{itemize}

%% file: conclusion.tex
\section{Conclusion}\label{sec:conclusion}

We close with the questions that the results leave open.

\paragraph{Is the transfer assumption true?}
Given the preparation assumption and rapid mixing, refuting the transfer assumption is the same as exhibiting a forger (\Cref{thm:general}), and that forger would have to use the coherence of a banknote issued by the mint, since forgers that measure first then fail (\Cref{thm:collapse-equiv}).  The only use of that coherence known for invariant money is to forward the banknote as one of the two outputs.  If the preparation assumption holds, \Cref{thm:transfer-barrier} rules out deriving the transfer assumption with fully black-box reductions that call the forger only at the serial number they are given.  Whether reductions that also call the forger at serial numbers they mint can derive the transfer assumption is open.  For other families of states, a coherent copy is worth more than every measure-first strategy.  \emph{Quantum fire} consists of states that can be cloned from a coherent copy but cannot be efficiently rebuilt from any classical description~\cite{NZ24,BNZ25}, so for such states the analogue of the transfer assumption fails.  Such states exist unconditionally relative to a classical oracle~\cite{CGS25}.  No invariant money scheme is known whose banknotes behave this way.

\paragraph{Is orbit state preparation hard in the black-box model?}
Lutomirski's lower bound in a black-box model of the moves~\cite{Lutomirski11} is conditional on a query lower bound for orbit membership that is itself open~\cite[\S6]{Lutomirski11}.  With the moves given as an oracle, no unconditional query lower bound for orbit state preparation is known.

\paragraph{Does an unbroken scheme meet the design conditions?}
The two schemes for which this paper verifies all three design conditions, the graph scheme (\Cref{prop:graph-conditions} and \Cref{sec:coset}) and the linear scheme (\Cref{prop:design-insufficient}), are both broken.  Montgomery and Sharif establish the spectral gap of their scheme under a condition on its generating set for which they give heuristic evidence~\cite[Prop.~4.22, \S11]{MS24}.  For the knot scheme, the spectral gap of the move graph and the dispersion of the Alexander polynomial are both open.  Both are questions about the scheme's combinatorics, not about adversaries.

\paragraph{Does an assumption with classical inputs give quantum lightning?}
Lightning would give security without the transfer assumption, because the reduction from lightning to unforgeability lies outside the class of \Cref{thm:transfer-barrier}.  The obvious lightning analogue of the preparation assumption, that no efficient algorithm outputs a serial number together with a state that verification accepts there, fails for every scheme, since the mint is such an algorithm.  A useful analogue would have to restrict what the preparing algorithm sees, as the preparation assumption does by giving it only classical data.  A lightning adversary, however, keeps both of its outputs coherent, so such an assumption does not directly bound it.  Whether some assumption with classical inputs gives lightning is open.

%% file: acknowledgements.tex
\section*{Acknowledgements}

Anthropic Claude Opus 4.6, 4.8, 5, and 5.5 accessed in 2026 have been used to improve exposition of results, derive proofs of propositions and lemmas, check arguments systematically. The final exposition, proofs, and arguments have been edited and verified by the authors.

%% file: appendix.tex
 \appendix
\crefalias{section}{appendix}
\crefalias{subsection}{appendix}

\section{Design conditions for the graph scheme}\label{app:graph-conditions}

\begin{proposition}[Design conditions for the graph scheme]\label{prop:graph-conditions}
The graph scheme (\Cref{def:graph-scheme}) on $v$ vertices ($n=\binom v2$), with $r:=nv$, has dispersed serial numbers and large orbits, and it mixes rapidly.  Its spectral gap is at least $1/(v-1)$.
\end{proposition}
\begin{proof}
\emph{Dispersion.}  $I$ is injective on orbits, so each level set is one orbit, of at most $v!$ graphs, and
\[
    \mathrm{wt}(p) \;=\; \frac{|\calC_k|}{2^n} \;\le\; \frac{v!}{2^{\binom v2}} \;\le\; 2^{\,v\log v-\binom v2} \;=\; 2^{-\Omega(v^2)}\,,
\]
which is negligible in $n = \binom v2$.

\emph{Large orbits.}  Measuring a minted banknote yields a uniformly random graph $x$ (\Cref{def:measured-mint}).  The transpositions generate $\Sym(v)$, so the orbit of $x$ is $\{\pi\cdot x : \pi\in\Sym(v)\}$.  Since $\pi\cdot x = \sigma\cdot x$ exactly when $\sigma^{-1}\pi\in\mathrm{Aut}(x)$, the map $K\mapsto K\cdot x$ is a bijection from the left cosets of $\mathrm{Aut}(x)$ onto the orbit, which therefore has $v!/|\mathrm{Aut}(x)|$ elements.  So it has at most $q(n)$ elements only if $|\mathrm{Aut}(x)|\ge v!/q(n)$, which exceeds $1$ for large $n$, so
\[
    \Pr_x\bigl[|\calC_{k(x)}|\le q(n)\bigr] \;\le\; \Pr_x\bigl[\mathrm{Aut}(x)\ne\{1\}\bigr] \;=\; 2^{-\Omega(v)} \;=\; 2^{-\Omega(\sqrt n)}\,,
\]
the fraction of graphs with a non-trivial automorphism~\cite{ER63}.

\emph{Gap.}  Write $\mathcal T$ for the set of transpositions and $M$ for the non-lazy walk on $\Sym(v)$, $(Mf)(\pi) = |\mathcal T|^{-1}\sum_{\tau\in\mathcal T}f(\tau\pi)$.  Its largest non-trivial eigenvalue is $(v-3)/(v-1)$ and its smallest is $-1$~\cite{DS81}.  The lazy walk $\tfrac12(\openone+M)$ maps each eigenvalue $\lambda$ of $M$ to $(1+\lambda)/2$, so its eigenvalues lie in $[0,1]$ and
\[
    \max_{\text{non-trivial}}\;\tfrac12(1+\lambda) \;=\; \tfrac12\Bigl(1+\frac{v-3}{v-1}\Bigr) \;=\; \frac{v-2}{v-1} \;=\; 1-\frac1{v-1}\,.
\]
For a graph $x$, renaming by $\tau$ sends $\pi\cdot x$ to $(\tau\pi)\cdot x$, and left multiplication by $\tau$ sends each left coset $K$ of $\mathrm{Aut}(x)$ to the left coset $\tau K$.  So under the bijection $K\mapsto K\cdot x$ from the left cosets of $\mathrm{Aut}(x)$ onto the orbit, the walk on the orbit of $x$ is the walk on left cosets
\[
    (\bar Mf)(K) \;:=\; |\mathcal T|^{-1}\sum_{\tau\in\mathcal T}f(\tau K)\,.
\]
A function $f$ on cosets lifts to $\tilde f(\pi) := f(\pi\,\mathrm{Aut}(x))$, and
\[
    (M\tilde f)(\pi) \;=\; |\mathcal T|^{-1}\sum_{\tau\in\mathcal T}f\bigl(\tau\pi\,\mathrm{Aut}(x)\bigr) \;=\; (\bar Mf)\bigl(\pi\,\mathrm{Aut}(x)\bigr)\,, \qquad \langle\tilde f,\mathbf 1\rangle \;=\; |\mathrm{Aut}(x)|\,\langle f,\mathbf 1\rangle\,.
\]
If $\bar Mf = \lambda f$ with $f\perp\mathbf 1$, then $M\tilde f = \lambda\tilde f$ with $\tilde f\perp\mathbf 1$.  Every non-trivial eigenvalue of the walk on the orbit is therefore a non-trivial eigenvalue of the walk on $\Sym(v)$, and the same holds for the lazy walks $\tfrac12(\openone+\bar M)$ and $\tfrac12(\openone+M)$.  Hence $\Delta\ge1/(v-1)$, and $r\Delta\ge nv/(v-1)\ge n$.
\end{proof}

\section{The verification walk as an index set}\label{app:walk}

This appendix gives the details behind the remark in \Cref{sec:oracle} that the random walk of verification does not supply the missing erasure.  

\begin{proposition}[The verification walk as an index set]\label{prop:walk-index}
Fix a seed $x$, write $\calC:=\calC_{k(x)}$, and let $T\ge1$.  Let $W = [2|S|]^T$ be the sequences of $T$ indices of the move register (\Cref{sec:framework}), let $\varphi_x(w)$ be the result of applying them to $x$ in turn, the identity indices acting trivially, and let $\ket\Psi := |W|^{-1/2}\sum_{w\in W}\ket w\ket{\varphi_x(w)}$ be the $T$-step lazy walk from $x$ with its path recorded.  In parts~(i)--(iii), let $T:=\lceil 2n\ln2/\Delta\rceil$, and suppose the scheme mixes rapidly, so that $T$ is polynomial.
\begin{enumerate}[label=\textup{(\roman*)},leftmargin=*,nosep]
\item $(W,\varphi_x)$ is an index set for $\calC$ at $x$, and its fibres satisfy $|W_{x\to y}| = |W|\,(\hat B^T)_{yx}$ with
\[
    (1-2^{-n})\,\frac{|W|}{|\calC|} \;\le\; |W_{x\to y}| \;\le\; (1+2^{-n})\,\frac{|W|}{|\calC|}\qquad\text{for every } y\in\calC\,.
\]
\item There are a set $F\subseteq W$ and relabellings $\tau_y$ of $W$ with $\braket{F}{\tau_y(W_{x\to y})}\ge1-2^{-n+1}$ for every $y\in\calC$.
\item If relabellings $\tau_y$ and a set $F$ satisfy $\braket{F}{\tau_y(W_{x\to y})}\ge\varepsilon$ for every $y\in\calC$, then the element register of $\Gamma\ket\Psi$, with $\Gamma$ as in \Cref{lem:uncompute}, has accepted weight at least $\varepsilon^2(1-2^{-n})$.
\item For every $T$, and whatever the spectral gap, the element register of $\ket\Psi$ has accepted weight exactly $1/|\calC|$.
\end{enumerate}
\end{proposition}
\begin{proof}
(i) Each index is uniform among $2|S|$ values, and $|S|$ of them are identity indices, so a uniformly random $w\in W$ performs $T$ steps of the lazy walk of equation~\eqref{eq:markov-generic}.  Hence $|W_{x\to y}|/|W| = (\hat B^T)_{yx}$, the probability that the walk from $x$ ends at $y$.  By the spectral decomposition of $\hat B$ on $\calC$ and the choice of $T$,
\[
    \Bigl|(\hat B^T)_{yx}-\frac1{|\calC|}\Bigr| \;\le\; (1-\Delta)^T \;\le\; e^{-\Delta T} \;\le\; 2^{-2n} \;\le\; \frac{2^{-n}}{|\calC|}\,,
\]
using $|\calC|\le2^n$.  This gives the fibre bounds, and in particular every fibre is non-empty, so $\varphi_x(W)=\calC$.

(ii) Let $F\subseteq W$ have $\min_y|W_{x\to y}|$ elements, and for each $y$ let $\tau_y$ be a relabelling with $\tau_y(W_{x\to y})\supseteq F$, which exists because $|F|\le|W_{x\to y}|$.  Then, by part~(i),
\[
    \braket{F}{\tau_y(W_{x\to y})} \;=\; \frac{|F|}{\sqrt{|F|\,|W_{x\to y}|}} \;=\; \Bigl(\frac{|F|}{|W_{x\to y}|}\Bigr)^{1/2} \;\ge\; \Bigl(\frac{1-2^{-n}}{1+2^{-n}}\Bigr)^{1/2} \;\ge\; 1-2^{-n+1}\,.
\]

(iii) By \Cref{lem:uncompute}, $\Gamma\ket\Psi = \sum_y c_y\ket{\tau_y(W_{x\to y})}\ket y$ with $c_y = (|W_{x\to y}|/|W|)^{1/2}\ge(1-2^{-n})^{1/2}|\calC|^{-1/2}$ by part~(i).  Its overlap with $\ket F\otimes\ket{\psi_{k(x)}}$ is
\[
    \sum_{y\in\calC}c_y\,|\calC|^{-1/2}\,\braket{F}{\tau_y(W_{x\to y})} \;\ge\; (1-2^{-n})^{1/2}\,|\calC|^{-1}\sum_{y\in\calC}\varepsilon \;=\; (1-2^{-n})^{1/2}\,\varepsilon\,.
\]
The element register's accepted weight is at least $\bra{\psi_{k(x)}}\rho\ket{\psi_{k(x)}}$ for its reduced state $\rho$, which is at least the squared overlap, $\varepsilon^2(1-2^{-n})$.

(iv) Grouping by endpoint,
\[
    \ket\Psi \;=\; \sum_{y\in\calC}\bigl((\hat B^T)_{yx}\bigr)^{1/2}\ket{W_{x\to y}}\ket{y}\,,
\]
with the $\ket{W_{x\to y}}$ orthonormal, because the fibres are disjoint.  The reduced state of the element register is therefore the diagonal $\rho=\sum_y(\hat B^T)_{yx}\ket y\!\bra y$, and its accepted weight is
\[
    \bra{\psi_{k(x)}}\rho\ket{\psi_{k(x)}} \;=\; |\calC|^{-1}\sum_{y\in\calC}(\hat B^T)_{yx} \;=\; \frac1{|\calC|}\,,
\]
since $\rho$ is supported on $\calC$ and each column of $\hat B^T$ sums to $1$.
\end{proof}

\paragraph{What the proposition shows.}
Part~(iv) is the accepted weight of $\ket x$ alone, so recording the path does not help.  Montgomery and Sharif prove the analogous bound for a superposition over curves entangled with class-group elements~\cite[Thm.~9.7]{MS24}.  Parts~(ii) and~(iii) show that the obstacle is computational.  Relabellings that nearly meet the label condition exist, but no polynomial-time family of them is known, and concatenation supplies none, since the index sequences of a fixed length do not form a group acting on the orbit.  A polynomial-time family with $\varepsilon\ge1/\poly(n)$ would, by part~(iii) and \Cref{lem:amplification}, solve orbit state preparation.

\section{Subexponential forgery of the class-group scheme}\label{app:elliptic-curves}

In the scheme of Montgomery and Sharif, the objects are elliptic curves with square-free Frobenius discriminant.  The class group $W = \mathrm{Cl}(\mathcal O)$ of the common endomorphism ring acts on each isogeny class $\calC$ by $(c,E)\mapsto c\ast E$ with trivial stabilizers.  Each move is the action of a class of small prime norm, and under the generalized Riemann hypothesis these classes generate $W$, so $\calC$ is a single orbit~\cite[\S3.3]{MS24}.  With trivial stabilizers, $\varphi_E : c\mapsto c\ast E$ is a bijection from $W$ onto $\calC$, so its inverse is the only section, and computing it is the group-action discrete logarithm: given $E$ and $c\ast E$, find $c$.  Condition~\condH{} is not known to hold, because the action of an arbitrary class is not known to be computable in polynomial time.  The best known algorithms take subexponential time~\cite{CJS14}, \cite[\S1.5]{Zhandry24}.  Montgomery and Sharif describe the construction of \Cref{prop:coset-osp} for their scheme, and note that solving the discrete logarithm would forge it~\cite[\S2.2, \S9.2]{MS24}.  The next proposition makes this explicit with the best known algorithms.

\begin{proof} [Proof of \Cref{prop:ms-forgery}]
    Measure the banknote, obtaining a curve $E$ in its isogeny class $\calC$.  The class group $W = \mathrm{Cl}(\mathcal O)$ of the common endomorphism ring acts on $\calC$ by $(c,E')\mapsto c\ast E'$ with trivial stabilizers.  Each move is the action of a class of small prime norm, and under the generalized Riemann hypothesis these classes generate $W$, so $\calC$ is a single orbit~\cite[\S3.3]{MS24}.  Curves over a finite field are isogenous exactly when they have the same number of points, by Tate's theorem, so the level set of $p$ is $\calC$ and $\ket{\$_p} = \ket{\psi_{k(E)}}$.  Let $\varphi_E(c) := c\ast E$.  Under the generalized Riemann hypothesis, the subroutines of the chain in the proof of \Cref{prop:coset-osp} are available, each run coherently and with error reduced to $2^{-\Omega(n)}$ on every input:
\begin{itemize}[leftmargin=*,nosep]
\item the action $\varphi_E(c) = c\ast E$, classically in time $2^{o(n)}$~\cite{CJS14};
\item the section $u_E(E') = $ the unique $c$ with $c\ast E = E'$, quantumly in time $2^{o(n)}$~\cite{CJS14};
\item the preparation of $\ket W$, in quantum polynomial time~\cite[\S1.5]{Zhandry24}.
\end{itemize}
Since stabilizers are trivial, $\mathrm{Stab}(E) = \{1_W\}$, and the chain maps
\[
    \ket W\ket0\ket0 \;\longmapsto\; \ket{1_W}\otimes\ket{\psi_{k(E)}}\otimes\ket0 \;=\; \ket{1_W}\otimes\ket{\$_p}\otimes\ket0\,,
\]
up to $2^{-\Omega(n)}$, because it makes $O(1)$ calls to the subroutines.  It runs in time $2^{o(n)}$.  Running it twice gives registers $\rho_1,\rho_2$, each within trace distance $2^{-\Omega(n)}$ of $\ket{\$_p}$.  Each therefore has accepted weight at least $1-2^{-\Omega(n)}$ and passes with at least that probability (\Cref{lem:spectral-overlap}), so both pass with probability at least $(1-2^{-\Omega(n)})^2 = 1-2^{-\Omega(n)}$.
\end{proof}

\section{Proof of \texorpdfstring{\Cref{lem:amplification}}{the amplification lemma}}\label{app:amplification}

The proof is fixed-point amplitude amplification between two rotations: one about the state that $D$ outputs, and one about the states whose element register holds the orbit state, which the coherent verifier implements up to a small error.

\begin{proof}[Proof of \Cref{lem:amplification}]
Write $\ket\Omega:=D\ket x\ket{0^m}$, write $\openidws$ for the identity on the ancilla register, and let $\Pi':=\ket{\psi_k}\!\bra{\psi_k}\otimes\openidws$ and $\mathcal V:=\mathrm{span}\{\ket\Omega,\Pi'\ket\Omega\}$.  By hypothesis~(a), $\aw_p(\sigma)=\bra{\psi_k}\sigma\ket{\psi_k}$, so~(b) says $\|\Pi'\ket\Omega\|^2\ge\eta$.

\emph{The coherent verifier.}  Let $r' := \lceil\Delta_\ast^{-1}\ln(1/\xi)\rceil$ and $K := \hat B^{r'}\Pi_{P_p}$.  The coherent verifier $V_{\mathrm{test}}$ computes $I$ into a fresh register and runs $r'$ rounds of the mixing test, each on a fresh copy of the move register, without measuring.  It sets a flag qubit if $I\ne p$ or if a round rejects.  On an input $\ket\phi$ with fresh ancillas it therefore acts as
\[
    V_{\mathrm{test}}\ket\phi\ket0 \;=\; \ket{a_{\mathrm u}}\otimes K\ket\phi + \ket{f_\phi}\,,
\]
where $\ket{a_{\mathrm u}}$ is the fixed ancilla state of the unflagged branch and $\ket{f_\phi}$ lies in the flagged subspace.  Since $\Delta\ge\Delta_\ast$,
\[
    \|K-\Pi_1^{(p)}\| \;\le\; (1-\Delta)^{r'} \;\le\; e^{-\Delta_\ast r'} \;\le\; \xi\,,
\]
and $K^\dagger K = \Pi_{P_p}\hat B^{2r'}\Pi_{P_p}$ has eigenvalue $1$ on the range of $\Pi_1^{(p)}$ and eigenvalues at most $\xi^2$ elsewhere.

\emph{Rotation about the target.}  Let $\Phi_\alpha$ apply the phase $e^{i\alpha}$ on the unflagged branch.  Then
\[
    V_{\mathrm{test}}^\dagger\,\Phi_\alpha\,V_{\mathrm{test}}\ket\phi\ket0 \;=\; \ket\phi\ket0 - (1-e^{i\alpha})\,V_{\mathrm{test}}^\dagger\bigl(\ket{a_{\mathrm u}}\otimes K\ket\phi\bigr)\,.
\]
The component of $V_{\mathrm{test}}^\dagger(\ket{a_{\mathrm u}}\otimes K\ket\phi)$ with fresh ancillas is $K^\dagger K\ket\phi\otimes\ket0$, since $\bra\chi\bra0V_{\mathrm{test}}^\dagger(\ket{a_{\mathrm u}}\otimes K\ket\phi) = \bra\chi K^\dagger K\ket\phi$ for every $\ket\chi$.  The remaining component has squared norm $\|K\ket\phi\|^2-\|K^\dagger K\ket\phi\|^2 = \bra\phi(K^\dagger K-(K^\dagger K)^2)\ket\phi \le \xi^2$.  Since also $\|K^\dagger K-\Pi_1^{(p)}\|\le\xi^2$, for every phase $\alpha$ the operator $V_{\mathrm{test}}^\dagger\Phi_\alpha V_{\mathrm{test}}$ is within $2(\xi+\xi^2)=O(\xi)$ in operator norm, on inputs with fresh ancillas, of the rotation $\openone-(1-e^{i\alpha})(\Pi_1^{(p)}\otimes\openidws)$. If $I$ is injective on orbits, $\Pi_1^{(p)}=\ket{\psi_k}\!\bra{\psi_k}$.  Otherwise, $\Pi_1^{(p)}$ restricts to $\ket{\psi_k}\!\bra{\psi_k}$ on the block of $\calC_k$ (\Cref{lem:eigenspace}), which contains $\mathcal V$ by hypothesis~(a). Either way the rotation acts on $\mathcal V$ as the rotation about the range of $\Pi'$.

\emph{Rotation about $\ket\Omega$.}  For every phase $\alpha$,
\[
    R_\Omega(\alpha) \;:=\; \openone-(1-e^{i\alpha})\ket\Omega\!\bra\Omega \;=\; D\bigl(\openone-(1-e^{i\alpha})\ket{x,0^m}\!\bra{x,0^m}\bigr)D^\dagger\,,
\]
which uses $D$, $D^\dagger$ and a copy of $x$.  It preserves $\mathcal V$, since $\ket\Omega\in\mathcal V$.

\emph{Amplification.}  Write $R_\tau(\beta) := \openone-(1-e^{i\beta})\Pi'$ for the rotation about the target, and $\ket\tau := \Pi'\ket\Omega/\|\Pi'\ket\Omega\|$.  Since $\|\Pi'\ket\Omega\|^2\ge\eta$, fixed-point amplitude amplification~\cite{YLC14} gives phases $\alpha_1,\beta_1,\dots,\alpha_N,\beta_N$, depending only on $\delta$ and $\eta$, with $N = O(\eta^{-1/2}\log(1/\delta))$ and
\[
    \bigl\|R_\tau(\beta_N)R_\Omega(\alpha_N)\cdots R_\tau(\beta_1)R_\Omega(\alpha_1)\ket\Omega - e^{i\theta}\ket\tau\bigr\| \;\le\; \delta
\]
for some phase $\theta$.  Every factor preserves $\mathcal V$, because $\Pi'\mathcal V\subseteq\mathcal V$.  The circuit $\mathrm{Amp}$ applies this product with each $R_\tau(\beta_j)$ replaced by $V_{\mathrm{test}}^\dagger\Phi_{\beta_j}V_{\mathrm{test}}$, which by the previous step is within $O(\xi)$ of $R_\tau(\beta_j)$ on $\mathcal V$ with fresh ancillas.  By the triangle inequality, with $\xi := \delta\eta^{1/2}/\log(1/\delta)$,
\begin{align*}
    \bigl\|\mathrm{Amp}\ket x\ket0 - e^{i\theta}\ket\tau\bigr\| &\;\le\; \delta + N\cdot O(\xi) \;=\; \delta + O\bigl(\eta^{-1/2}\log(1/\delta)\bigr)\cdot\frac{\delta\eta^{1/2}}{\log(1/\delta)} \;=\; O(\delta)\,,\\
    r' &\;=\; \bigl\lceil\Delta_\ast^{-1}\ln(1/\xi)\bigr\rceil \;=\; O\bigl(\Delta_\ast^{-1}\log(1/\delta\eta)\bigr)\,.
\end{align*}
The normalized target is
\[
    \frac{\Pi'\ket\Omega}{\|\Pi'\ket\Omega\|} \;=\; \ket{\psi_k}\otimes\ket{a_x}\,,\qquad \ket{a_x} \;:=\; \frac{(\bra{\psi_k}\otimes\openidws)\ket\Omega}{\|(\bra{\psi_k}\otimes\openidws)\ket\Omega\|}\,.
\]
The phase $e^{i\theta}$ is absorbed into $\ket{a_x}$.  Rescale $\delta$ by a constant. The copies of $\mathsf A$ are ancillas of the final circuit, which is uniform and depends on $x$ only through its input.
\end{proof}